\documentclass[journal,letterpaper]{IEEEtran}
\usepackage{hyperref}
\usepackage{amsmath,amssymb,amsthm}
\usepackage{graphics, subfigure}
\usepackage{tikz}
\usepackage{pgfplots}
\usepackage{algorithm, algorithmic}
\usepackage{color}
\usepackage{multirow}
\usetikzlibrary{automata,positioning}
\usepackage{wrapfig}
\usepackage{rotating}
\usepackage{flushend}
\usepackage{bbm}

\usepackage{algorithm, algorithmic}
\newtheorem{remark}{Remark}
\newtheorem{theorem}{Theorem}

\newtheorem{lemma}{Lemma}
\usepackage{epstopdf, xspace}

\usepackage{mathtools}

\DeclarePairedDelimiter\floor{\lfloor}{\rfloor}

\usepackage{pgfplots} 
\pgfplotsset{compat=newest} 
\pgfplotsset{plot coordinates/math parser=false} 
\newlength\figureheight 
\newlength\figurewidth

\newcommand{\ie}{\textit{i.e.,}\xspace}
\newcommand{\eg}{\textit{e.g.,}\xspace}

\begin{document}
\definecolor{brown}{cmyk}{0,0.81,1,0.60}
\definecolor{magenta}{rgb}{0.4,0.7,0}
\definecolor{gray}{rgb}{0.5,0.5,0.5}
\definecolor{red}{rgb}{1,0,0}
\definecolor{green}{rgb}{0.5,0,0.5}
\definecolor{blue}{rgb}{0,0,1}


\ifthenelse{\isundefined{\final}} {
\newcommand{\vaneet}[1]{{\color{green}(VA: #1)}}
\newcommand{\shuai}[1]{{\color{red}(SH: #1)}}
\newcommand{\feng}[1]{{\color{blue}(FQ: #1)}}
\newcommand{\anis}[1]{{\color{brown}(AE: #1)}}
\newcommand{\shubho}[1]{{\color{magenta}(SS: #1)}}
}{
\newcommand{\vaneet}[1]{}
\newcommand{\shuai}[1]{}
\newcommand{\feng}[1]{}
\newcommand{\anis}[1]{}
\newcommand{\shubho}[1]{}
}

\newcommand{\eat}[1]{}

\newcommand{\BULLET}{\vspace{+.03in} \noindent $\bullet$ \hspace{+.00in}}

%

\title{Optimized Preference-Aware Multi-path Video Streaming with
Scalable Video Coding }


 \author{Anis Elgabli, Ke Liu, and Vaneet Aggarwal \thanks{The authors are with Purdue University, West Lafayette, IN 47907, email:\{aelgabli, liu2067, vaneet\}@purdue.edu}}
 
\maketitle

\begin{abstract}

Most client hosts are equipped with multiple network interfaces (e.g., WiFi and cellular networks). Simultaneous access of multiple interfaces can significantly improve the users' quality of experience (QoE) in video streaming. An intuitive approach to achieve it is to use Multi-path TCP (MPTCP). However, the deployment of MPTCP, especially with link preference, requires OS kernel update at both the client and server side, and a vast amount of commercial content providers do not support MPTCP.  Thus, in this paper, we realize a multi-path video streaming algorithm in the application layer instead, by considering Scalable Video Coding (SVC), where each layer of every chunk can be fetched from only one of the orthogonal paths. We formulate the quality decisions of video chunks subject to the available bandwidth of the different paths, chunk deadlines, and link preferences as an optimization problem. The objective is to to optimize a QoE metric that maintains a tradeoff between maximizing the playback rate of every chunk and ensuring fairness among chunks. The QoE is a weighted some of the following metrics: skip/stall duration, average playback rate, and quality switching rate. However, the weights are chosen such that pushing more chunks to the same quality level is more preferable over any other choice. Even though the formulation is a  non-convex discrete optimization, we show that the problem can be solved optimally with a polynomial complexity in some special cases. We further propose an online algorithm where several challenges including bandwidth prediction errors, are addressed. Extensive emulated experiments in a real testbed with real traces of public dataset reveal the robustness of our scheme and demonstrate its significant performance improvement compared to other multi-path algorithms. 


\end{abstract}




\section{Introduction}
\label{sec:intro}
	It is common that today's client hosts are equipped with multiple network interfaces. For example, mobile devices (\eg Apple iOS 7 \cite{ios_mptcp}) inherently support WiFi and cellular networks at the same time. The provision of simultaneous access of the multiple interfaces significantly improves the performance of various applications, \eg web browsing \cite{han2015anatomy}, as they can leverage the bandwidths of both the links (or paths). Video streaming is one of the major sources of  traffic in mobile networks. While its popularity is on the rise, its quality of experiences (QoE) is still often far from satisfactory. In this paper, we propose a set of efficient video streaming algorithms to improve users' QoE using multiple paths simultaneously. In multi-path video streaming, one of the links is in general preferable as compared to the other. For instance, the users may not wish to use too much of cellular link since it is, in many cases, more expensive (limited plans) and less energy efficient (far from the base station). Therefore, the scenario where some of the links are less preferable than others need to be considered.

%
%


An intuitive approach to enable multi-path is to replace the conventional transport (\ie TCP) with Multi-path TCP (MPTCP~\cite{mptcp}), which is the de-facto multi-path solution allowing applications to transparently use multiple paths. Specifically, MPTCP opens multiple sub-flows (usually one over each path), distributes the data onto the sub-flows at the sender, and reassembles data from each path at the receiver. The key advantage of MPTCP is that it allows applications to use multiple paths without changing the existing socket programming interface. Despite these advantages, the deployment of MPTCP is sluggish. A vast amount of commercial content providers do not support MPTCP~\cite{chen2016msplayer,deconinck16:pam} because it requires OS kernel update at both the client and server side. To make things worse, MPTCP uses special TCP extensions that are often blocked by middle-boxes  of the commercial content providers (e.g MPTCP over Port 80/443 is blocked by most U.S. cellular carriers~\cite{ashkan16_mobicom}). Implementing MPTCP with link preference further requires message exchange between the rate adaptation logic at the application layer and  MPTCP in order to disable/enable parallel TCP connections. Thus, in this paper, we will consider an approach for fetching video encoded using scalable video coding (SVC) on multiple paths without the use of MPTCP.  


There are two popular coding techniques, Advanced Video Coding (AVC, e.g MPEG4-AVC) and Scalable Video Coding (SVC~\cite{SVC_encoding}).
In AVC, each video chunk is stored into $L$ {\em independent} encoding versions. When fetching
a chunk, 
the player's adaptation mechanism, Adaptive Bit Rate (ABR) streaming, needs to select one out of the $L$ versions based on its estimation of the network condition and the buffer capacity. 

In SVC, each chunk is encoded into ordered \emph{layers}: one \emph{base layer} (Layer 0) with the lowest playable quality, and multiple \emph{enhancement layers} (Layer $i>0$) that further improve the chunk quality based on layer $i-1$.  For decoding a chunk up to enhancement layer $i$, a player must download all layers from 0 to $i$. Thus, adaptive SVC streaming can allow playback at a lower quality if all the enhancement layers have not been fetched while ABR streaming does not allow playback if the chunk is not fully downloaded. Adaptive SVC streaming has been shown to provide better adaptiveness and scalability than ABR \cite{SVCMetrics,Sanchez12}. That is why we choose SVC as the coding scheme in this paper and propose a set of streaming algorithms using SVC.


Instead of using MPTCP, we realize ``multi-path" in application layer, \ie initiating a separate connection in the application (\eg browser, video player) via one of the network interfaces/link (\eg WiFi and LTE) using conventional TCP, and each layer of a chunk is fetched using one of the connections. Thus, the streaming algorithm decides whether to fetch a layer of the chunk or not, and which link to use for fetching. This approach requires no change to the server side and is compatible with any middle-box. We note that the different layers of a chunk can be fetched using different links using adaptive SVC streaming, while the entire chunk is fetched from the same link when using ABR. This flexibility helps to provide additional improvement on QoE compared to ABR-based mechanisms. As shown in Section \ref{sec:eval}, our SVC-based multi-path streaming algorithms outperform  one of the recent state-of-the-art ABR-based multi-path streaming algorithms, MSPlayer \cite{chen2016msplayer}.

We consider two classes of streaming algorithms: {\em skip based} and {\em no-skip based} streaming.
 The former is for real-time streaming: each chunk is associated with a deadline, chunks not received by their  deadlines are skipped. For no-skip based streaming, if a chunk cannot be received by its deadline, it will not be skipped; instead, a stall (re-buffering) will incur until it is fully received. For both the scenarios,  we formulate the adaptive streaming algorithm as an optimization problem for perfectly predicted bandwidths of the available links. The objective is to maximizes the video quality and minimize the stalls/skips simultaneously (c.f., Section 3) while respecting the link preference, bandwidth and chunk deadlines constraints. Even though the formulation is a non-convex optimization problem, we can show that the optimal solution can be achieved in polynomial time for some practical cases. In practice, the future bandwidth cannot be perfectly predicted, but can be estimated for a smaller window ahead using a crowd-sourced method to obtain historical data \cite{Riiser12,GTube}, or harmonic mean of the bandwidth achieved for the last few seconds \cite{chen2016msplayer}. Therefore, we propose an online sliding based algorithm that solves the proposed optimization problem every $\alpha$ seconds to make a decision for the next $W$ chunks. The perfect prediction case forms an upper bound to the algorithm performance with imperfect prediction. Our proposed online adaptation streaming algorithms incorporate the imperfect information by using the bandwidth prediction methods in~\cite{Riiser12,GTube,chen2016msplayer} (these prediction methods are not perfect thus have errors) as shown in the online algorithms in Section 3 and 4. We also show that MPTCP is a special case of our algorithm which represents the scenario that combines the WiFi and LTE links as a single link that has the total bandwidth of both links. Moreover, for MPTCP, the proposed algorithm works for both adaptive SVC and ABR schemes.

{\bf Our Contributions:} The main contributions of the paper are as follows.

\BULLET We  formulate the multi-path SVC video streaming with perfect  bandwidth prediction  as an optimization problem, whose objective is to maximize the users' quality of experience (QoE). We consider two classes of streaming algorithms:  skip based and  no-skip based streaming. For both algorithms, the goal of the scheduling algorithm is to determine up to which layer we need to fetch for each chunk (except for those skipped in realtime streaming), such that the overall playback bitrate is maximized and the number of stalls or skipped chunks is minimized without violating link preference, bandwidth, and chunk deadlines constraints.  We also propose an online algorithm for the scenario where  bandwidth prediction is not perfect, \ie available for short period ahead and has errors.

\BULLET The proposed problem is a non-convex discrete optimization problem. There are discrete variables and  non-convex constraints. However, we develop an efficient algorithm that solves this specific problem in polynomial time and shown to be optimal for some practical cases. Thus, we provide a class of discrete optimization problem that is solvable optimally in polynomial time under certain assumptions. 
Specifically, we solve the proposed integer-constrained problem using  an easy-to-solve packing based algorithm.


\BULLET We also propose special cases of our approach, \ie a set of single path adaptive streaming algorithms using MPTCP for all the above cases, \ie with and without preference, and skip-based or no-skip-based. These algorithms can be used on either of ABR or adaptive SVC streaming schemes.

\BULLET  {We evaluated our algorithms using a TCP/IP test bed with real SVC encoded videos and bandwidth traces from public datasets collected from commercial  networks. The evaluation demonstrates that our approach is robust to prediction errors, and works well with a short prediction window, where we estimate the bandwidth using harmonic mean of the bandwidth values of the past few seconds. We evaluated our algorithms against a number of adaptive streaming strategies including the multi-path version of the buffer-based approach (MP-BBA) \cite{BBA} and the prediction based algorithms such as MSPlayer~\cite{chen2016msplayer}.} 
The results demonstrate that our algorithm outperforms them by improving the key QoE metrics such as the playback quality, the number of layer switches, and the number of skips or stalls. For example, our skip based streaming algorithm was able to achieve average playback quality that is $25\%$, and $35\%$ higher than MP-BBA and MSPlayer respectively with lower stall/skip durations. The preference-aware adaptive streaming algorithms were compared with the preference-aware MPTCP based algorithms in \cite{han2016mp} and it is shown that the proposed algorithm obtains lower skips, higher average quality, and lower link $1$ usage thus demonstrating  improvement in all these metrics.



The rest of the paper is organized as follows. Section \ref{sec:related} discusses the related work. Section \ref{sec:problem} describes the problem formulation. In section \ref{skipalgo} and \ref{noSkip}, a set of polynomial run time algorithms are provided for solving the non convex problem considering Skip and No-Skip scenarios respectively. Moreover, optimality is shown for some special cases. Section \ref{sec:eval} presents the trace-driven evaluation results with comparison to the different baselines. Section \ref{sec:concl} concludes the paper. 



\section{Related Work}
\label{sec:related}

Video streaming has received a lot of attention from both the academia and industry in the past decade.
There are ABR and adaptive SVC adaptation algorithms. Some of the widely used ABR streaming techniques include  MPEG-DASH \cite{DASH}, Apple's HLS~\cite{HLS}, Microsoft's Smooth Streaming~\cite{SS}, and Adobe's HDS~\cite{HDS}. In recent studies,  various approaches for making ABR streaming decisions have been investigated, for example, by using control theory~\cite{MPC,Miller15}, Markov Decision Process~\cite{Jarnikov11}, machine learning~\cite{Claeys14}, client buffer information~\cite{BBA}, and data-driven techniques~\cite{Liu12,C3,CS2P}. 
 
  SVC received the final approval to be standardized as an amendment of the H.264/MPEG-4 AVC (Advanced Video Coding) standard
in 2007~\cite{svcrelease}. Although
much less academic research has been conducted on SVC compared to AVC-style schemes over regular H.264,
there exist some studies of using SVC to adapt video playback quality to network conditions.
A prior study~\cite{LayeredAdaptation} proposed a server-based quality adaptation mechanism that performs coarse-grained rate adaptation by adding or dropping layers of a video stream.
While this mechanism was designed to be used over UDP with a TCP-friendly rate control, more recent research has explored techniques that use SVC over HTTP.
A study~\cite{Sanchez12} compared SVC with regular H.264 encoding (H.264/AVC).
Their results suggest SVC outperforms AVC
for scenarios such as VoD and IPTV through more effective rate adaptation.
The work~\cite{SVCDataset} published the first dataset and toolchain for SVC.
Some prior work~\cite{Andelin12,Sieber13} proposed new rate adaptation algorithms for SVC that prefetch future base layers and backfill current enhancement layers. Even though optimization-based formulations have been proposed for video streaming in the past \cite{MPC,6115763,Ahmedin:2014:ESV:2741897.2741942,FENG2017}, the optimality guarantees for the proposed algorithms are limited. In contrast, this paper shows the optimality of the proposed algorithm for a non-convex discrete optimization problem.

%



%
The knowledge of the future network conditions can play an important role in Internet video streaming.
%
%
A prior study~\cite{Zou15} investigated the performance gap between state-of-the-art streaming approaches and the approach with accurate bandwidth prediction.
The results indicate that prediction brings additional performance boost for ABR streaming, and thus motivates our study. The bandwidth have been shown to be predictable for some time ahead using a crowd-sourced method to obtain historical data \cite{Riiser12,GTube}, or harmonic mean of the past bandwidth \cite{chen2016msplayer}. 
%

Adaptive streaming strategies have been proposed for multi-path channels in~\cite{chen2016msplayer} where a heuristic based on prediction, MSPlayer,  was proposed for streaming AVC video using WiFi and LTE. In their proposed heuristic, alternate chunks are downloaded using WiFi and LTE connections respectively. The key differences with this work are: 1) We use the flexibility of SVC, where the different layers can be fetched over different paths, and 2) The proposed algorithm is shown to be optimal for some special cases. Recently, the authors of \cite{han2016mp} gave novel algorithms for using multi-path TCP to stream AVC videos where the  primary objective was to reduce the usage of LTE as well as minimizing the stall. The approach in \cite{han2016mp} uses a rate adaptation algorithm, like BBA  \cite{BBA} or Festive \cite{jiang2012improving}, and uses MPTCP to fetch AVC videos at the same quality levels while reducing the usage of LTE. This approach works on top of rate adaptation techniques, which do  not explicitly minimize LTE usage in their objective. In contrast, this paper propose algorithms that consider preference of one link over the other explicitly and use the approach of  \cite{han2016mp}  as a comparison. Moreover,  \cite{han2016mp} considered use of MPTCP and a no-skip based version. However, in this paper we consider both skip and no-skip based scenarios, as well as both options of using or not using MPTCP. The versions that use MPTCP in this paper can also be used directly using AVC rather than SVC since they do not exploit fetching a layer from only one of the links.


\section{Skip Based Streaming: Problem Formulation}
\label{sec:problem}

In this section, we describe our problem formulation considering two links (\eg WiFi and cellular networks) in the exposition, but the formulation and proposed algorithms can be easily extended to  more links. In skip based streaming, the video is played with an initial start-up, (\textit{i.e.,} buffering) delay $s$ and there is a playback deadline for each of the chunks where chunk $i$ need to be downloaded by time $deadline(i)$. Chunks not received by their respective deadlines are skipped. The objective of the proposed formulation to find the fetching policy (which link should fetch which of the chunk layers) such that the number of skipped chunks is minimized as the first priority and overall playback bitrate is maximized as the next priority without violating link preference, bandwidth, and chunk deadlines constraints.


We first assume that the future bandwidth is perfectly known beforehand, the buffer capacity is infinite, and each layer is encoded at constant bit rate (CBR). In other words, all chunks have the same $n$th layer size. We will relax both the perfect prediction and the infinite buffer size assumptions later on in this section. Moreover, we will evaluate the proposed algorithm with videos that are encoded at Variable Bit Rates (VBR). With these assumptions, we give a formulation for skip based video streaming. Let us assume a video divided into $C$ chunks (segments), where every chunk is of length $L$ seconds, is encoded in Base Layer (BL) with rate $r_0$ and $N$ enhancement layers ($E_1, \cdots, E_{N}$) with rates $r_1, \cdots, r_{N}$,  respectively. Note that $Y_n=L*r_n$ is the size of the $n$th layer. $Z_{n,i}$ denotes the size of the $n$th layer that has been fetched. Therefore, if the $n$th layer can be fetched $Z_{n,i}=Y_n$; otherwise $Z_{n,i}=0$.

 Let $z_n^{(1)}(i,j)$ be the size of layer $n$ of chunk $i$ fetched over the first link (e.g., LTE link) at time slot $j$, and $z_n^{(2)}(i,j)$ be the size of  the layer $n$ of chunk $i$ over the second link (e.g., WiFi link) at time slot $j$. Moreover, let $Z_{n,i}^{(k)}$ be the total size that is fetched for the layer $n$ of chunk $i$ over the link $k$. \ie $Z_{n,i}^{(k)}=\sum_{j=1}^{(i-1)L+s} z_n^{(k)}(i,j)$. Let $B^{(k)}(j)$ be the available bandwidth over the  link $k\in \{1,2\}$ at time $j$ and $s$ be the startup delay. As mentioned, for the time being we assume the bandwidth can be perfectly predicted and we will relax this assumption in Section \ref{sec:bw_err}. We assume all time units are discrete and the discretization time unit is assumed to be 1 second. Finally, we define the decision variable of layer $n$, chunk $i$ and link $k$ ($I_{n,i}^{(k)}$) as follows:
 
 \begin{equation}\label{equ:fetchingPolicy}
\left\{\begin{array}{l}
I_{n,i}^k=1, \text{ if the $n$-th layer of chunk $i$ is fetched by link $k$}\\
I_{n,i}^k=0,  \text{ otherwise}\\
\end{array}\right.
\end{equation}

We assume that link 1 can be used as much as its bandwidth allows. However, we assume that link 2 can only help in fetching up to the layer $n_2 \leq N$ if link 1 can't meet the deadline. Note that if $n_2=N$, this is a special case in which both links can be used equally likely to fetch all layers without any preference. In the other extreme, when $n_2=0$, link 2 can be used only to avoid skips if link 1 fails to do so. 

We start by assuming that the decision is taken at the application layer where a layer of a chunk  cannot be split over the two paths (or links). It must be fully downloaded over either of the paths.  In other words, $Z_{n,i}^{(1)} \cdot Z_{n,i}^{(2)}=0$ for all $n$ and $i$. The key objectives of the problem are \emph{(i)} minimization of  the number of skipped chunks, \emph{(ii)} maximization of the average playback rate of the video,  and \emph{(iii)} minimization of  the quality changes between the neighboring chunks to ensure that perceived quality is smooth.

In order to account for these objectives and respect the priority order, we maximize a weighted sum of layer decision variables of the two links, where lower layers and more preferable links are given higher weights. We introduce weights $\lambda_n^{k}$ where $n$ is the layer index and $k$ is the link index. The weights need to satisfy the following condition.

 \begin{equation}
 \lambda_a^{(1)} > C \cdot \Big( \sum_{n=a}^{N}\lambda_n^{(2)}+ \sum_{n=a+1}^{N}\lambda_n^{(1)}\Big) .
 \label{basic_gamma_1_0}
\end{equation}
 \begin{equation}
 \lambda_a^{(2)}  > C \cdot \Big( \sum_{n=a+1}^{N}\lambda_n^{(2)}Y_n\Big) .
 \label{basic_gamma_1_1}
\end{equation}


The choice of $\lambda$'s that satisfies Equation \ref{basic_gamma_1_0} and \ref{basic_gamma_1_1} implies two requirements. First, it implies that, for any layer $a$, the layers higher than  $a$ have lower utility than a chunk at layer $a$. In the case when $a=0$, the choice of $\lambda$ implies that all the enhancement layers achieve less utility than one chunk at the base layer. The use of $\lambda$ helps in giving higher priority to pushing more chunks to the $n^{th}$ layer quality over fetching some at higher quality at the cost of dropping the quality of other chunks to below the $n^{th}$ layer quality. Second, the choice of $\lambda$ implies that the highest utility that can be achieved for every layer is when it is fetched over link 1. This will obey the priority order of the links, and it will not use a less preferable link to fetch a layer that can be fetched by a more preferable link.

\begin{eqnarray}
\textbf{Maximize: } \sum_{i=1}^C\sum_{n=0}^{N}\Big(\lambda_n^1I_{n,i}^{(1)}+\lambda_n^2I_{n,i}^{(2)}\Big) 
\label{equ:eq1}
\end{eqnarray}
subject to
\begin{eqnarray}
\sum_{j=1}^{(i-1)L+s} z_n^{(k)}(i,j) = Z_{n,i}^{(k)}\quad  \forall i,  n, k \in \{1,2\} 
\label{equ:c0eq1}
\end{eqnarray}
\begin{eqnarray}
I_{n,i}^{(1)}\cdot Z_{n,i}^{(1)}+I_{n,i}^{(2)}\cdot Z_{n,i}^{(2)}= Z_{n,i}\quad  \forall i,  n 
\label{equ:c1eq1}
\end{eqnarray}
\begin{eqnarray}
Z_{n,i}\le \frac{Y_n}{Y_{n-1}}Z_{n-1,i}\quad  \forall i,  n>0 
\label{equ:c2eq1}
\end{eqnarray}
\begin{eqnarray}
\sum_{n=0}^N\sum_{i=1}^{C} z_n^{(k)}(i,j)  \leq B^{(k)}(j) \  \   \forall k\in\{1,2\}, \forall j,
\label{equ:c3eq1}
\end{eqnarray}
\begin{equation}
Z_{n,i}^{(1)} \cdot Z_{n,i}^{(2)}=0\ \forall i, n
\label{equ:c4eq1}
\end{equation}
\begin{equation}
z_n^{(k)}(i,j) \geq 0\   \forall k\in\{1,2\}, \forall i 
\label{equ:c5eq1}
\end{equation}
\begin{equation}
z_n^{(2)}(i,j) = 0\   \forall n > n_2, \forall i 
\label{equ:c6eq1}
\end{equation}
\begin{equation}
z_n^{(k)}(i,j)= 0\   \forall \{i: (i-1)L+s < j\}, k\in \{1,2\}
\label{equ:c7eq1}
\end{equation}
\begin{equation}
Z_{n,i} \in {\mathcal Z}_n\triangleq \{0, Y_n\} \quad  \forall i, n
\label{equ:c8eq1}
\end{equation}
\begin{equation}
I_{n,i}^{(k)} \in \{0, 1\} \quad  \forall i, n,k
\label{equ:c9eq1}
\end{equation}
\vspace{-.2in}
\begin{eqnarray}
\text{Variables:}&& z_n^{(2)}(i,j), z_n^{(1)}(i,j),  Z_{n,i} \ \ \  \forall   i = 1, \cdots, C, \nonumber \\&& j = 1, \cdots, (C-1)L+s, n = 0, \cdots, N \nonumber
\end{eqnarray}

In the above formulation, constraints \eqref{equ:c0eq1} ensures that the total amount fetched by link $k$ over all times to be equal to $Z_{n,i}^{(k)}$ . Constraints  \eqref{equ:c1eq1} and \eqref{equ:c8eq1} ensure that what is fetched for layer $n$ of chunk $i$ over all links and times to be either zero or $Y_n$. The constraint  \eqref{equ:c2eq1} enforces not to fetch the $n$th layer of any chunk if the lower layer is not fetched.
\eqref{equ:c3eq1}  imposes the bandwidth constraint of the two links  at each time slot $j$.
Constraint \eqref{equ:c4eq1} enforces a layer of a chunk to be fetched only over one of the paths. Constraint \eqref{equ:c5eq1} imposes the non-negativity of the download of a chunk. Constraint \eqref{equ:c6eq1} enforces that link 2 can't be used to fetch layers higher than the layer $n_2$. \eqref{equ:c7eq1} imposes the deadline constraint since chunk $i\in \{1, \cdots, C\}$ cannot be fetched after its deadline ($deadline(i)=(i-1)L+s$). Recall that $s$  is the initial startup delay. Constraint\eqref{equ:c8eq1} enforces that a layer is either completely fetched or completely skipped (No partial fetching of the layers). Finally, constraint \eqref{equ:c9eq1} enforces the decision variable of the $n$-th layer, $i$-th chunk and $k$-th link to be either $0$ or $1$.

\subsection{Structure of the Proposed Problem}

The problem defined in~\S\ref{sec:problem} has integer constraints and a non-convex constraint (in \eqref{equ:c5eq1}). Integer-constrained  problems are in the class of discrete optimization \cite{nemhauser1988integer}. Some of the problems in this class are the Knapsack problem, Cutting stock problem, Bin packing problem, and Traveling salesman problem. These problems are all known to be NP hard. Very limited problems in this class of discrete optimization are known to be solvable in polynomial time, some typical examples being shortest path
trees, flows and circulations, spanning trees, matching, and matroid
problems. The well known Knapsack problem optimizes a linear function with a single linear constraint ( for integer variables), and is known to be NP hard. The optimization problem defined in this paper has multiple constraints (including non-convex constraints), and does not lie in any class of known  problems that are polynomially-time solvable to the best of our knowledge. 

We note that optimization based formulations have been proposed earlier for video streaming \cite{MPC,6115763,Ahmedin:2014:ESV:2741897.2741942,FENG2017}. However,  the optimality of these algorithms have not been considered due to integer constraints. The non-convex constraint in \eqref{equ:c5eq1} further increases the problem complexity beyond the proposed ILP formulations in \cite{6115763,Ahmedin:2014:ESV:2741897.2741942}.   In this paper, we propose a set of polynomial complexity algorithms to solve the proposed optimization problem, and we show that optimal solution of the proposed algorithm can be achieved for some special and practical assumptions. In particular we show that optimal solution can be achieved when one of the links can only be used to avoid skips/stalls.
\section{MultiPath SVC Streaming Algorithms}\label{skipalgo}
 In this section, we describe the proposed algorithm for skip-based streaming. To develop the algorithm, we first consider the offline algorithm in which the bandwidth is perfectly known for the whole period of the video. We first describe the algorithm when both links can be used without preference, \ie  $\lambda_n^1=\lambda_n^2, \forall n$. We refer to this algorithm as ``Offline Multi-Path SVC Algorithm" (Offline MP-SVC). This algorithm is, then, extended to the Preference-Aware case in which link 1 is more preferable and link 2 can only be used to fetch up to a certain layer. We refer to this algorithm as Offline Pref-MP-SVC. Moreover, we propose ``Avoid-Skips MP-SVC'', an efficient algorithm for the special case of link preference in which link 2 can only be used to avoid skips (\ie, $n_2=0$) , and we show that ``Avoid-Skips MP-SVC'' (\ie, $n_2=0$) achieves the optimal solution of the proposed algorithm.  Finally, we propose an online algorithm (Online Pref/No-Pref MP-SVC) in which more practical assumptions are considered, such as short bandwidth prediction with prediction error and finite buffer size.
\subsection{MP-SVC}\label{skipalgoEqual}
 In this section, we describe our algorithm (Offline MP-SVC), which is summarized in
Algorithm \ref{algo:mpsvc}. The algorithm initially calculates the cumulative bandwidth of every second $j$ and link $R^{(k)}(j)$ (Line 5). Then, it makes the decision for the base layer, \ie which chunks to be skipped and which to be fetched. The algorithm performs forward scan and finds the maximum number of base layers that can be fetched before the deadline of every chunk $i$ ($V_{0,i}$). The maximum number of base layers that can be fetched before the deadline of the $i^{th}$ chunk is: $V_{0,i}=\sum_{k=1}^2 \floor{\frac{R^k(deadline(i)}{r_0}}$(Line 6). Let $skip(i)$ be the total number of skips before the deadline of the chunk numbered $i$. Therefore, if $V_{0,i}$ is less than $i-skip(i-1)$ at the deadline of the $i^{th}$ chunk, there must be a skip/skips, and the total number of skips from the start until the deadline of chunk $i$ will be equal to $skip(i)=skip(i-1)+1$(lines 9-11). If there are $A$ skips, the algorithm will always skip the first $A$ chunks since they are the closest to their deadlines. Thus, skipping them will result in a bandwidth that can be used by all of the remaining chunks to increase their quality to the next layer. This choice maximizes the total available bandwidth for the later chunks. 
Before we describe the second step into details, we define some parameters. Let $\alpha_i^n(k)$ be the amount of bandwidth used to fetch the layer $n$ of chunk $i$ by link $k$ before the deadline of chunk numbered $i-1$. Let $\zeta_i^n(k)$ be the cost of fetching the layer $n$ of chunk $i$ by link $k$, and $\zeta_i^n(k)$ can be found as follows:

\begin{equation}\label{equ:cost}
\left\{\begin{array}{l}
\zeta_{i,n}(k)=\alpha_i^n(k), \text{if $Y_n$ can be fetched by link $k$}\\
\zeta_{i,n}(k)=\infty,  \text{ otherwise}\\
\end{array}\right.
\end{equation}

With $\zeta_i^n(k)$ being defined, we describe the second step of the algorithm. The algorithm performs backward scan per chunk starting from the first chunk that was decided to be fetched by calling Algorithm 2 (line 18). The backward scan simulates fetching every chunk $i$ starting from its deadline and by every link. The algorithm computes the the cost of fetching the base layer of chunk $i$ by every link $k$ ($\zeta_i^n(k)$). The link choice that minimizes the cost is chosen to fetch the base layer of chunk $i$. Note that the link over which the chunk $i$ will be fetched is the one that gives the maximum amount of total available bandwidth over all links before the deadline of chunk numbered $i-1$.  For example, consider that fetching the $i^{th}$ chunk by link 1 results in using  $x$ amount of the bandwidth before the deadline of $i-1^{th}$ chunk while fetching the $i^{th}$ chunk by link 2 results in using  $y$ amount of the bandwidth before the deadline of $i-1^{th}$ chunk. Then, the first link will be chosen to fetch the chunk $i$ if $x<y$. The objective is to free as much as possible of early bandwidth since it will help more chunks to fetch their higher layers because it comes before their deadlines. 


{\bf Enhancement layer modifications}: The algorithm proceeds in performing forward-backward scan per enhancement layer in order. The bandwidth is now modified to be the remaining bandwidth after excluding whatever has been reserved to fetch lower layers (Line 6 in Algorithm 2).  Moreover, we don't shift the deadline of a chunk if it can't receive its $n$-th layer before its deadline that was found by running the base layer forward scan; instead, we skip fetching the $n$-th layer of that chunk. \ie, this chunk is not a candidate to the $n$-th layer quality. Also note that $n^{th}$ layer for a chunk is not considered if its $n-1^{th}$ layer is not decided to be fetched (Lines 9 and 16 in Algorithm 1).

\begin{figure}
    \scalebox{0.95}{
		\begin{minipage}{\linewidth}
			\begin{algorithm}[H]
			\small
				\begin{algorithmic}[1]
				\STATE {\bf Input:}  $\mathbf{Y}=\{Y_n\forall n\}$, $L$, $s$, $C$, $\mathbf{B}^{(k)}=\{B^{(k)}(j)\forall j\}$, $k=1, 2$. 
				\STATE {\bf Output:} $I^{(k)}_n$, $k=1, 2$: set containing the indices of the chunks that can have their $n$th layer fetched over link $k$.	
								
   \STATE  $deadline(i)=(i-1)L+s \quad \forall i$
  
    \FOR{each layer $n = 0, \cdots, N$}   
     \STATE $R^{(k)}(j)=\sum_{j^\prime=1}^{j} B^{(k)}(j^\prime)$,$\forall j,k$
	\STATE $V_{n,i}=\sum_{k=1}^2\floor{\frac{R^{(k)}(deadline(i)}{r_n}}$, $\forall i$
    \STATE $skip(0)=0$
     \FOR{$i=1:C$}
      \IF{$V_{n,i} < i-skip(i-1)$ or ($n \neq 0$ and $i \notin  I_{(n-1)}^{(1)}$ and $i \notin I_{(n-1)}^{(2)}$)}
      \STATE $skip(i)=skip(i-1)+1$
      \ENDIF
      \ENDFOR
        \STATE Skip the first ``$skip(C)$'' chunks.
        \STATE $i^\prime=skip(C)+1$: the index of the first chunk to fetch
     
    \FOR{$i=i^\prime:C$}
    \STATE $j_k =deadline(i),k=\in \{1, 2\}$
    	\IF{($n=0$ or $i \in I_{n-1}^{(1)}$ or $i \in I_{n-1}^{(2)})$}
				\STATE ${B2}_{k}={B}_{k}, \forall k$, $t=deadline(i-1)$
      				\STATE $[B2_{k},\zeta_i^n(k)]=$ Backward$(k,i, j_k, \mathbf{B2}_{k}, Y_n, t) \forall k$
			\STATE $k_1=\arg\min (\zeta_i^n)$
      				\STATE $I_n^{(k_1)}=I_n^{(k_1)} \cup i$, $\mathbf{B}_{(k_1)}=B2_{(k_1)}$

	\ENDIF
      \ENDFOR
     \ENDFOR
   				\end{algorithmic}
				\caption{Offline MP-SVC Algorithm }\label{algo:mpsvc}
			\end{algorithm}
		\end{minipage}
	}
	\end{figure}

\begin{figure}
		\vspace{-.1in}
		\begin{minipage}{\linewidth}
			\begin{algorithm}[H]
				\small
				\begin{algorithmic}[1]
				\STATE {\bf Input:} $k,i, j, B2, Y_n, t$
					\STATE {\bf Output:} $\zeta_i^n(k)$the cost of fetching layer $n$ of chunk $i$ by link $k$, $B2$ is the residual bandwidth after fetching chunk $i$.
					
    \STATE {\bf Initialization:}
   $\zeta_k(i,n)=0$
     \WHILE {($Y_n > 0$)}
            \STATE $fetched=min(B2(j), Y_n)$
             \STATE $B2(j)=B2(j)-fetched$, $Y_n(i)=Y_n(i)-fetched$,
            \STATE {\bf if } {$(j \leq t)$} {\bf then }  $\zeta_i^n(k)=\zeta_i^n(k)+fetched$
	   \STATE {\bf if } {$(B2(j)=0)$} {\bf then }  $j=j-1$            	
	    \STATE {\bf if } {$j < 1$ and $Y_n > 0$} {\bf then }  $\zeta_i^n(k)=\infty$, break
\ENDWHILE
				\end{algorithmic}
				\caption { Backward Algorithm }
			\end{algorithm}
		\end{minipage}
		\vspace{-.2in}
	\end{figure}

{\bf Chunks Download: }During the actual fetching of the chunks. Each link fetches the layers in order of the chunks they belong to. In other words, the $n$-th layer of chunk $i$ is fetched before the $m$-th layer of chunk $j$ if $i < j$. Moreover, for the same chunk, the layers are fetched according to their order. For example, the base layer is the first layer to download since if it is not received by the chunk's deadline, the chunk will not be played. Moreover, non of the higher layer received can be decoded if the base layer is not received by the chunk's deadline.

{\bf Complexity Analysis}:  The algorithm sequentially decides  each layer one after the other, and thus it is enough to find the complexity of one layer to compute overall algorithm complexity. 
For each layer $n$, the algorithm first finds the cumulative bandwidth of every time slot $j$, $R^{(k)}(j)$ which is linear complexity. Then, it starts from the last chunk, and performs  the backward algorithm on each link at each time. The backward algorithm determines whether the chunk can be fetched in that link and the amount of residual bandwidth. Since the complexity of backward algorithm is linear in $C$, the overall complexity for a layer is $O(C^2)$. Thus, the overall complexity is $O(NC^2)$.

\begin{figure*}
\centering
\includegraphics[trim=0.6in 1.5in 0.6in 1.5in, clip,  width=\textwidth]{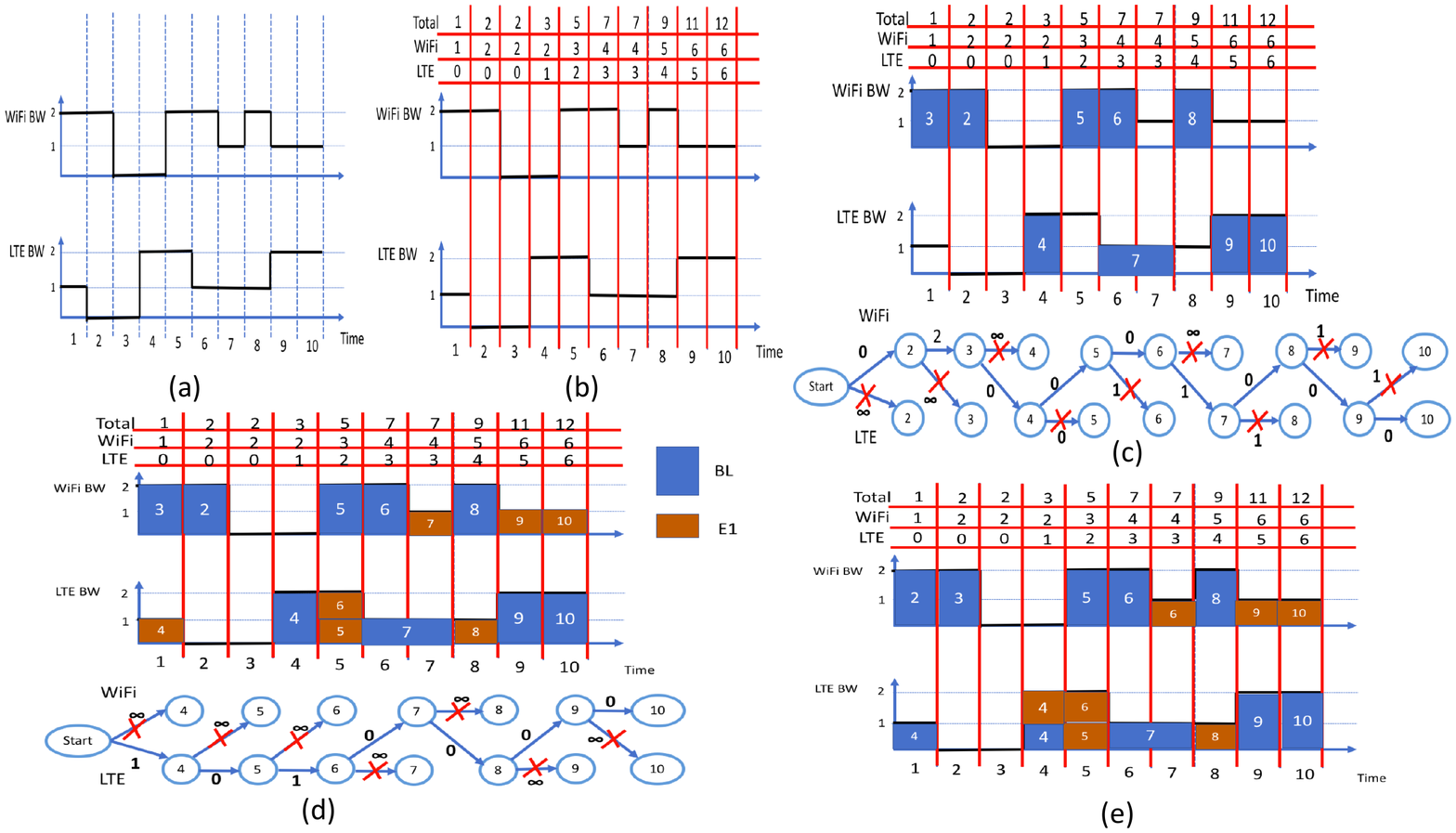}	
 \vspace{-.3in}
 \caption{ Illustration Example of MP-SVC Algorithm: (a) Bandwidth Trace, (b) Base Layers that can be fetched before the deadline of every chunk, (c) Base Layer Decision, (d) $1^{st}$ Enhancement Layer Decision, and (e) The Actual Fetching Policy}
 \label{fig:ex1}
\end{figure*}

{\bf Example: } Fig~\ref{fig:ex1} shows an example that illustrates how MP-SVC algorithm works. We assume a video that consists of 10 chunks, 1 second length each. The video is an SVC encoded into 1 Base Layer($BL$) and 1 Enhancement Layer($E_1$). The $BL$ and $E_1$ sizes are $2Mb$, and $1Mb$ respectively. \ie  $Y_0=2Mb$, and the $Y_1=1Mb$. Moreover, we assume that the startup delay is 1 second. Therefore, the $deadline(i)=i, \forall i$. 

We assume a mobile device with two links Link 1 is WifI link and link 2 is LTE link. Fig~\ref{fig:ex1}-a shows the bandwidth traces of the two paths (WiFi and LTE), and Fig~\ref{fig:ex1}-b show the result of the first forward scan. The forward algorithm finds the maximum number of base layers that can be fetched before the deadline of every chunk. We clearly see that up to the deadline of the $3$rd chunk, only 2 chunks can be fetched. Therefore, one out of the first 3 chunks should be skipped. The algorithm as explained previously decides to skip the first chunk (chunk 1) since the bandwidth that the first chunk leaves can be available to all of the remaining chunks for the next layer decisions. In other words, if we skip chunk 3, then it is possible that part or the whole bandwidth that it leaves comes after the deadline of chunks 1 and 2 which means that chunks 1 and 2 can't benefit from this bandwidth in fetching their higher layers.

Forward scan finds the chunks that can have their base layers fetched without violating their deadlines. Consequently, Backward scan described in Fig~\ref{fig:ex1}-c finds the fetching policy such that all these chunks have their base layers fetched without violating their deadlines as promised by forward algorithm. Moreover, backward algorithm finds the policy that maximizes the total bandwidth of every chunk for the next layer decision. Before, describing Fig~\ref{fig:ex1}-c, let's denote the cost of fetching a chunk $i$ over link $k$ by $c_i^k$. Moreover $c_i^k$ is equal to the size of the portion of chunk $i$ that is fetched before the deadline of chunk $i-1$. Therefore, the link that has less $c_i^k$ is chosen to fetch chunk $i$. Note that $k=1$ represents the WiFi link

 As shown in Fig~\ref{fig:ex1}-c, the algorithm simulates fetching every chunk starting from its deadline back. Moreover, it does not consider chunk 1 since forward algorithm decided that chunk 1 can't be fetched. The algorithm proceeds in the order of the chunks, finds $c_i^k, k \in \{1,2\}$, and decides which link should fetch which chunk.  For example, chunk  $4$ can totally be fetched after the deadline of chunk 3 over LTE link, so $c_4^2=0$. In the other hand, it can't even be fetched over the WiFi chunk unless an earlier chunk is skipped, so $c_4^1=\infty$. Therefore, chunk $4$ is assigned to the LTE link.
 
 Fig~\ref{fig:ex1}-d shows how the algorithm repeats the same process described in Fig~\ref{fig:ex1}-c but for the $1^{st}$ enhancement layer decisions. In the $E1$ decisions, the algorithm uses the remaining bandwidth of each link after excluding whatever reserved for fetching the $BL$s. Moreover, since the $BL$ of the first chunk is not decided to be fetched, its $E1$ is not considered. Finally, Fig~\ref{fig:ex1}-e shows the actual fetching of the chunks. As shown in the figure, each link fetches the layers in order of the chunks they belong to, and for the same chunk, the layers are fetched according to their order. For example, In LTE link, $E1$ of the $4$-th chunk is fetched before $E1$ of the $6$-th chunk.

\begin{remark}\label{rem_mptcp}
We note that if one can use Multi-path TCP, the above algorithm can be used by assigning the sum of the bandwidths of the two links to one path and having the bandwidth of the other path as zero. 
\end{remark}

\begin{remark}\label{rem_mptcp}
The decision of combined MP-SVC and MPTCP can also be used for AVC encoded videos as well since after the decisions are taken, chunks can be fetched in order. Therefore, in this case, the decision variable $Z_{n,i}$ represent the size difference between the $n$-th and $(n-1)$-th quality levels of chunk $i$. 
\end{remark}

\begin{remark} 
The algorithm can be easily  extended to the case when there are more than 2 links where  each layer of a chunk is assigned to the link that leaves more total residual bandwidth to the earlier chunks.
\end{remark}

\begin{lemma}
Given size decisions up to $n-1^{th}$ layer ($Z_{0,i}, \cdots,Z_{n-1,i}$, for all $i$), remaining bandwidth, and $deadline(i)$ for every chunk $i$, the proposed algorithm achieves the minimum number of $n^{th}$ layer skips (or obtains the maximum number of chunks at layer $n$) as compared to any feasible algorithm which fetches the same layers of every chunk up to layer $n-1$.
 \label{them:skip1} 
\end{lemma}
\begin{proof}
Proof is provided in Appendix \ref{proof_skip}.
\end{proof}

\begin{lemma}
Given size decisions up to $n-1^{th}$ layer ($Z_{0,i}, \cdots,Z_{n-1,i}$, for all $i$), running the proposed algorithm for the $n$-th layer decisions provides the maximum total remaining bandwidth for every chunk $i$ before its deadline for decisions of layers $> n$ as compared to any feasible algorithm which fetches the same layers of every chunk up to layer $n-1$.
 \label{them:skip2} 
\end{lemma}
\begin{proof}
Proof is provided in Appendix \ref{proof_skip}.
\end{proof}

{\bf MPTCP-SVC: } We conclude this section by mentioning that the MPTCP variant of our problem in which the bandwidth of the two links is aggregated to create a virtual link 1 and have the other link's bandwidth zero all the time is a special case of the proposed problem in \cite{AnisSinglePath}. In particular it is the case in which $\beta=1$ and the buffer is infinite. Moreover, our MPTCP algorithm (MPTCP-SVC) can easily be shown that it does not perform worse than DBPBP algorithm which means that the MPTCP-SVC achieves the optimal solution of the proposed algorithm. We omit the proofs of the optimality of MPTCP-SVC since this paper does not assume that a layer can be split over multiple links. We just use MPTCP-SVC for the sake of performance comparison.

\subsection{Pref-MP-SVC}\label{sec:pref}

In this section, we consider a preference to the use of link $1$ as compared to link $2$. We first assume that there is a parameter $n_2\in \{0, \cdots, N\}$ indicating that the link $2$ should not be used to obtain chunks above layer $n_2$. This will allow using link $2$ to help getting lower layers while not over-using it to achieve higher layers. Further, among different schemes that obtains same number of chunks till layer $n2$, we wish to minimize the usage of link $2$. For $n_2=0$, this implies that link $2$ is only used to avoid skips and not to fetch any enhancement layers. Further, link $2$ is used only if necessary to reduce skips, when link $1$ is not sufficient to obtain the optimal number of base layer chunks. 

%
%
%
%

The Pref-MP-SVC algorithm is described by Algorithm \ref{pref-algo}. In the first step, we use the MP-SVC algorithm for layers $n=0$ to $n_2$. Then, we wish to minimize the usage of Link $2$ such that the decisions till layer $n_2$ remain the same. In order to do that, we re-run MP-SVC, but we consider only link 1 (link 2 bandwidth is zero) and chunks that were initially decided to be fetched by link 2. We exclude the bandwidth that was reserved to fetch chunks over link 1 from the previous run of MP-SVC. The main objective of the $2^{nd}$ run of MP-SVC is to reduce the usage of link 2 to its minimum but with ensuring that maximum number of chunks can be fetched at least at $n_2$ layer quality. 

\begin{figure}
	\vspace{-.1in}
 \scalebox{0.95}{	\begin{minipage}{\linewidth}
		\begin{algorithm}[H]
			\small
			\begin{algorithmic}[1]
					\STATE {\bf Input:}  $\mathbf{Y}=\{Y_{n,i}\forall n,i\}$, $L$, $s$, $C$, $n_2$, $\mathbf{B}_{k}$, $k=1,2$. 
				\STATE {\bf Output:} $I_n^{(k)}$, $k=1,2$: set containing the indices of the chunks that can have their $n$th layer fetched by link $k$.	
				\STATE $[I_n^{(k)},\mathbf{B}_{k}, n=0\cdots n_2]=$offline MP-SVC($\mathbf{Y}$, $L$, $s$, $\mathbf{B}_{k}$)
				\STATE $\mathbf{Y2}=\{Y_{n,i}\forall n,i \in I_n^{(2)}\}$		
				
				\STATE $[I_n^{(1)},\mathbf{B}_{1}, n=0\cdots n_2]=$offline MP-SVC($\mathbf{Y2}$, $L$, $s$, $\mathbf{B}_{1}$)
				
				\FOR{$i=1:C$}
					\IF{$(i \in I_n^{(1)}\}$ and $i \in I_n^{(2)}\})$}
						\STATE $I_n^{(2)}=I_n^{(2)}-\{i\}$
						\ENDIF
				\ENDFOR
				
			\STATE $[I_n^{(1)},\mathbf{B}_{1}, n=n_2+1\cdots N]=$offline MP-SVC($\mathbf{Y_n}, n>n_2$, $L$, $s$, $\mathbf{B}_{1}$)
			\end{algorithmic}
			\caption{Offline Pref MP-SVC Algorithm }
			\label{pref-algo}
		\end{algorithm}
	\end{minipage}
}
\end{figure}

\subsection{Avoid-Skips MP-SVC}\label{sec:pref}

A special case link preference scenario is when link 2 can only be used to avoid skips. \ie  link 2 can only help in fetching base layers. We propose an efficient algorithm to solve the problem in this specific case, and we show that the proposed algorithm is optimal. We name this algorithm ``{\bf Avoid Skips MP-SVC}''. In fact ``Avoid-Skips MP-SVC'' modifies the MP-SVC decisions to account for the facts that link 2 usage should be at its minimum, link 2 can't be used to fetch beyond the base layer, the available bandwidth for next layer decision over link 1 is maximized.

Avoid Skips MP-SVC works as follows: In the first step, the algorithm finds the minimum number of base layer skips similar to the MP-SVC algorithm. Consequently, the algorithm initially assumes that all non skipped chunks can be fetched by link 1, and only assign base layers to link 2 if a deadline of a chunk $i$ is reached and $V_{1,i} < i-skip(i)$. In this case, the algorithm moves to the less preferable link (link 2) the earliest base layers that can be fetched by this link (lines 16-24). Moving the earliest possible such that all non skipped chunks can be fetched will maximize the remaining bandwidth on link 1 before the deadline of every chunk. Remember, link 2 will not be used to fetch beyond the base layer, so only link 1 is used to fetch higher layers. Therefore, moving earlier chunks to link 2 can have the highest number of candidate chunks to the next layer decision.
\begin{figure}
    \scalebox{0.95}{
		\begin{minipage}{\linewidth}
			\begin{algorithm}[H]
			\small
				\begin{algorithmic}[1]
				\STATE {\bf Input:}  $\mathbf{Y}=\{Y_n\forall n\}$, $L$, $s$, $C$, $\mathbf{B}^{(k)}=\{B^{(k)}(j)\forall j\}$, $k=1, 2$. 
				\STATE {\bf Output:} $I^{(k)}_n$, $k=1, 2$: set containing the indices of the chunks that can have their $n$th layer fetched over link $k$.	
								
   \STATE  $deadline(i)=(i-1)L+s \quad \forall i$

    \FOR{each layer $n = 0, \cdots, N$} 
  \STATE $R^{(k)}(j)=\sum_{j^\prime=1}^{j} B^{(k)}(j^\prime)$,$\forall j,k$
	\STATE $V_{n,i}^{(1)}=\floor{\frac{R^{(1)}(deadline(i)}{r_n}}$, $V_{n,i}^{(2)}=\floor{\frac{R^{(k)}(deadline(i)}{r_n}}$, $\forall i$
	\STATE $V_{n,i}=V_{n,i}^{(1)}+V_{n,i}^{(2)}$, $\forall i$
	
	\STATE $skip(0)=0$
     \FOR{$i=1:C$}
       \IF{$V_{n,i} < i-skip(i-1)$ or ($n \neq 0$ and $i \notin  I_{(n-1)}^{(1)}$ and $i \notin I_{(n-1)}^{(2)}$)}
      \STATE $skip(i)=skip(i-1)+1$
      \ENDIF
      \ENDFOR
        \STATE Skip the first ``$skip(C)$'' chunks.
        \STATE $i^\prime=skip(C)+1$: the index of the first chunk to fetch
\IF{$n=0$}	
	\STATE $M_1=skip+1:C$, $M2=[]$
	\FOR{$i=1:C$}
		\WHILE{($V_{1,i} < i-skip(i)$ and $V_{2,i} \geq  i-skip(i)$)}
			\STATE $move=min(M)$
			\STATE $M_1=M_1-\{move\}$, $M_2=M_2\cup\{move\}$

		\ENDWHILE
	\ENDFOR
\ENDIF

    \FOR{$i=i^\prime:C$}
    \STATE $j_k =deadline(i),k=\in \{1, 2\}$

    	\IF{$n=0$ or $i \in I_{n-1}^{(1)}$ or $i \in I_{n-1}^{(2)})$}
\IF{$n=0$}
\STATE {\bf if }{$i \in M_1$} {\bf then } ${B2}_{2}={\bf 0}$ {\bf else} ${B2}_{1}={\bf 0}$
\ELSE
	\STATE ${B2}_{1}={B}_{1}$, ${B2}_{2}={\bf 0}$
\ENDIF
\STATE $t=deadline(i-1)$
      				\STATE $[B2_{k},\zeta_i^n(k)]=$ Backward$(k,i, j_k, \mathbf{B2}_{k}, Y_n, t) \forall k$

			\STATE $k_1=\arg\min (\zeta_i^n)$

      				\STATE $I_n^{(k_1)}=I_n^{(k_1)} \cup i$, $\mathbf{B}_{(k_1)}=B2_{(k_1)}$

	\ENDIF
      \ENDFOR
     \ENDFOR
   				\end{algorithmic}
				\caption{Avoid-Skips MP-SVC Algorithm }\label{algo:mpsvc}
			\end{algorithm}
		\end{minipage}
	}
	\end{figure}

Now, we show that Algorithm \ref{pref-algo} achieves the optimal solution of the proposed formulation. 

\begin{lemma}
	Avoid-Skips MP-SVC Algorithm achieves the minimum number of base layer skips.
	\label{lem:avoid_skips}
\end{lemma}
\begin{proof}
Proof is provided in  Appendix \ref{proofPref0}. 
\end{proof}

\begin{lemma}\label{lem_lower}
	Among two strategies with the same number of chunks fetched at base layer, the one that obtains lower content over less preferable links is preferred.
	\label{lem:avoid_skips2}
\end{lemma}
\begin{proof}
Proof is provided in  Appendix \ref{proofPref1}. 
\end{proof}

\begin{lemma}
	Avoid Skips MP-SVC Algorithm has the minimum data usage of link 2. In other words, no other algorithm can achieve less data usage for link 2 with the same number of skips as that achieved by Avoid Skips MP-SVC.
	\label{lem:avoid_skips3}
\end{lemma}
\begin{proof}
Proof is provided in  Appendix \ref{proofPref2}. 
\end{proof}


\begin{lemma}
	Among all algorithms with the same number of base layer skips, Avoid Skips MP-SVC reserves the largest possible bandwidth over the preferred link for fetching the enhancement layers. In other words, the proposed algorithm maximizes the resources to the higher layers among all algorithms that have same base layer decisions.
	\label{lem:avoid_skips4}
\end{lemma}
\begin{proof}
Proof is provided in  Appendix \ref{proofPref3}. 
\end{proof}
Combining these results, the following theorem shows the optimality of the proposed algorithm. 
\begin{theorem}	
	Up to a given enhancement layer
	$M, M \geq 0$, if ${Z_{m}^{(k)}}^*$ is the size of the $m$th layer fetched by the $k$th link ($m \leq M$) of chunk $i$ that is found by running Avoid-Skips MP-SVC algorithm, and ${Z_{m}^{(k)}}^\prime$ is the size that is found by any other feasible algorithm such that all constraints are satisfied, then the following holds when $\lambda$'s satisfy \eqref{basic_gamma_1_0} and \eqref{basic_gamma_1_1}.
	\begin{align}
	& \sum_{i=1}^C\sum_{n=0}^{N}\Big(\lambda_n^1{I_{n,i}^{(1)}}^\prime+\lambda_n^2{I_{n,i}^{(2)}}^\prime\Big) \\& \leq  \sum_{i=1}^C\sum_{n=0}^{N}\Big(\lambda_n^1{I_{n,i}^{(1)}}^*+\lambda_n^2{I_{n,i}^{(2)}}^*\Big)
	\label{thm:thm_avoidSkips}
	\end{align}
	In other words, Avoid Skips MP-SVC finds the optimal solution to the optimization problem~(\ref{equ:eq1}-\ref{equ:c9eq1}) when $n_2=0$ and $\lambda_1$ and $\lambda_2$ satisfy \eqref{basic_gamma_1_0} and \eqref{basic_gamma_1_1}.
	
	\label{theorem: theorem2}
\end{theorem}
\begin{proof}
Proof is provided in  Appendix \ref{Prefproof}. 
\end{proof}

{\bf Pref-MPTCP-SVC: } In this algorithm, we first run MP-SVC algorithm for layers $n=0$ to $n_2$. However, we aggregate the bandwidth of the two links to form one link and we assume another link with the bandwidth set to zero all times. Once the decisions are found. We run a forward scan simulating fetching the chunks according to the decided quality using link 1 only. Any chunk can't meet its deadline over link 1, we use link 2 to fetch equivalent amount of the difference from the earliest chunk/chunks decided to be fetched by link 1. By this, we achieve the minimum usage of link 2 such that all MPTCP-SVC decisions are respected, and we also provide the maximum bandwidth over link 1 for next layer decisions.\ie, the bandwidth that is left in link 1 by moving earliest chunks to link 2 can be used to increase the quality of more chunks to the next layer since it comes before the deadline of more chunks than any other choice. MPTCP-SVC can be shown to be optimal and works for both AVC as well as SVC encoded videos. The optimality  proofs of Pref-MPTCP SVC are omitted since the paper is not for MPTCP scenario. We are just including MPTCP based algorithms for the sake of comparison.

\subsection{Online MP-SVC/Pref MP-SVC: Dealing with Short and Inaccurate BW Prediction}
\label{sec:bw_err}

For the algorithm described in Sections \ref{skipalgoEqual} and \ref{sec:pref}, we assumed that perfect bandwidth prediction, and  client buffer capacity is unlimited. However, practically, the prediction will not be perfect, and the client buffer might be limited. In this section, we will use an online algorithm that will obtain prediction for a window of size $W$ chunks ahead starting from the chunk that has the next time slot as its deadline and make decisions based on the prediction. Note that the buffer capacity at the user may be limited, and in that case we consider that $W$ is no more than the buffer capacity since the chunks beyond that will not be fetched in the algorithm. The way we capture the buffer capacity is by assuming that at every re-run of the algorithm (every $\alpha$ seconds), only chunks that are $B_{max}$ ahead of the chunk that is currently being played can be fetched.
 There are multiple ways to obtain the prediction, including a crowd-sourced method to obtain historical data \cite{Riiser12,GTube}. Another approach may be to use a function of the past data rates obtained as a predictor for the future, an example is to compute the harmonic mean of the past $\beta$ seconds to predict the future bandwidth ~\cite{chen2016msplayer}. The decisions are re-computed for the chunks that have not yet reached their deadlines periodically every $\alpha$ seconds.

For the prediction window $W$, the algorithm in Sections \ref{skipalgoEqual} and \ref{sec:pref} are run to find the quality using the predicted bandwidth profile, then the $W$ chunks are fetched according to the algorithm decision. If all $W$ chunks are fetched before the next re-computation time, the current time is set as the re-computation point, and the fetching policy for the next $W$ chunks from the one that is has the next time slot as its deadline is computed. Finally, at the start of the download, all links are assigned chunks at base layer quality to fetch since there is no bandwidth prediction available yet.

\subsection{No Skip Based  Streaming Algorithm}\label{no_skip}

In no-skip streaming (\textit{i.e.,} watching a pre-recorded video), when the deadline of a chunk  cannot be met,
rather than skipping it, the player will stall the video and continue downloading the chunk. The objective here is to maximize the average quality while minimizing the stall duration (the re-buffering time). The objective function is slightly different from equation~\eqref{equ:eq1} since we do not allow to skip the base layer. However, we skip higher layers.
For the constraints, all constraints are the same as skip based optimization problem except that we add constraint~(\ref{equ:c3q2}) to enforce $Z_{0,i}$ to be equal to the BL size of the chunk $i$ (it can't be zero). 
We define the total stall (re-buffering) duration from the start till the play-time of chunk $i$ as $d(i)$. Therefore, the deadline of any chunk $i$ is $(i-1)L+s+d(i)$.  The objective function is thus given as:
\begin{equation}
 \Big(\sum_{i=1}^C\sum_{n=0}^{N}\big(\lambda_n^1I_{n,i}^{(1)}+\lambda_n^2I_{n,i}^{(2)}\big)-\mu d(C)\Big)
 \label{equ:noSkipObj}
\end{equation}

where the weight for the stall duration is chosen such that $\mu \gg \lambda_n^{(1)}$, since users tend to care more about not running into re-buffering over better quality. 
This is a multi-objective optimization problem with quality and stalls as the two objectives, and is formulated as follows.
\begin{eqnarray}
\textbf{Maximize: } (\ref{equ:noSkipObj})
\label{equ:eq2}
\end{eqnarray}
subject to  \eqref{equ:c1eq1}, \eqref{equ:c2eq1}, \eqref{equ:c3eq1}, \eqref{equ:c4eq1} ,  \eqref{equ:c5eq1}, \eqref{equ:c6eq1}, \eqref{equ:c8eq1}, \eqref{equ:c9eq1},

\begin{eqnarray}
\sum_{j=1}^{(i-1)L+s+d(i)} z_n^{(k)}(i,j) = Z_{n,i}^{(k)}\quad  \forall i,  n, k \in \{1,2\} 
\label{equ:c0eq2}
\end{eqnarray}

\begin{equation}
z_n^{(k)}(i,j)= 0\   \forall \{i: (i-1)L+s+d(i) < j\}, k\in \{1,2\}
\label{equ:c1eq2}
\end{equation}
\begin{equation}
d(i) \geq d(i-1)\geq 0\   \forall i>0 \label{equ:c2eq2}
\end{equation}
\begin{equation}
Z_{0,i} =Y_0
\label{equ:c3q2}
\end{equation}

\begin{eqnarray}
\text{Variables:}&& z_n^{(2)}(i,j), z_n^{(1)}(i,j),  Z_{n,i}, d(i) \forall   i = 1, \cdots, C,  \nonumber \\
&& j = 1, \cdots, (C-1)L+s+d(C), n = 0, \cdots, N \nonumber
\end{eqnarray}

We note that this problem can be solved using an algorithm similar to that for the skip-based streaming. 
One difference as compared to the skip version is that the first step of the no-skip algorithm is to determine the minimum stall time since that is the first priority. Therefore, The No-Skip MP-SVC/ Pref-MP SVC/ Avoid-stalls MP SVC algorithms initially runs initial forward scan with the objective of checking if all chunks can be fetched at least at base layer quality with would be the minimum stall duration (the final startup delay for offline scenario). We will just explain how the MP-SVC is modified to account for the No-Skip Based scenario, and the same pre-processing steps are applied to both No-Pref MP-SVC and its special case No-Stall MP-SVC which replaces the No-Skip MP-SVC and follow the same optimality proof. The only difference is that the startup delay is chosen such that all chunks can be fetched at least at the base layer quality without running into skips.

In the first step, the The No-Skip MP-SVC algorithm runs a base layer forward scan that has the objective of checking if all chunks can be fetched at least at the base layer quality with the current startup delay.  At the deadline of any chunk $i$, if $V_{0,i} < i$, the algorithm increments the deadline of every chunk $ \geq i$ by $1$ and resumes the forward scan (Line 11). The algorithm does not proceed to the next chunk till the condition of $V_{0,i} \geq i$ is satisfied. At the end, the algorithm sets the final deadline of every chunk $i$ to be: $deadline(i)=(i-1)L+s+d(C)$. Therefore, the base layer forward scan achieves the minimum stall duration and brings all stalls to the very beginning since that will offer bandwidth to more chunks and can help increase the quality of more chunks. The detailed steps are described in Algorithm 5.




The rest of the algorithm is equivalent to skip version since skips are not allowed only for base layers (higher layers can be skipped). The detailed steps are given in Algorithm \ref{algo:no_skip}. The key difference in the no-skip version is that the startup delay is decided such that there will be no skips. Avoid-Skips algorithm is replaced with avoid-Stalls which works the same way as avoid skips after the initial forward scan that finds the minimum stall duration such that all chunks can be fetched at least at the base layer quality. Therefore, All the lemmas and theorems of skip version are applicable to no-skip version.



\begin{figure}
    \scalebox{0.95}{
		\begin{minipage}{\linewidth}
			\begin{algorithm}[H]
			\small
				\begin{algorithmic}[1]
				\STATE {\bf Input:}  $\mathbf{Y}=\{Y_n\forall n\}$, $L$, $s$, $C$, $\mathbf{B}^{(k)}=\{B^{(k)}(j)\forall j\}$, $k=1, 2$. 
				\STATE {\bf Output:} $I^{(k)}_n$, $k=1, 2$: set containing the indices of the chunks that can have their $n$th layer fetched over link $k$.	
								
   \STATE  $deadline(i)=(i-1)L+s \quad \forall i$
  
    \FOR{each layer $n = 0, \cdots, N$}   
     \STATE $R^{(k)}(j)=\sum_{j^\prime=1}^{j} B^{(k)}(j^\prime)$,$\forall j,k$
	\STATE $V_{n,i}=\sum_{k=1}^2\floor{\frac{R^{(k)}(deadline(i)}{r_n}}$, $\forall i$
   \IF{$n=0$}
       \FOR{$i=1:C$}
       \STATE {\bf If} $i>1$ {\bf then} $d(i)=d(i-1)$
       \WHILE{($V_{n,i} < i$)}
       		\STATE $d(i)=d(i)+1$
		\STATE $deadline(i)=(i-1)L+s+d(i)$
		\STATE $V_{n,i}=\sum_{k=1}^2\floor{\frac{R^{(k)}(deadline(i)}{r_n}}$
	\ENDWHILE
      \ENDFOR
	\STATE  $deadline(i)=(i-1)L+s+d(C), \quad \forall i$ , $i^\prime=1$
	\ELSE
		 \STATE $skip(0)=0$
     \FOR{$i=1:C$}
       \IF{$V_{n,i} < i-skip(i-1)$ or ($i \notin  I_{(n-1)}^{(1)}$ and $i \notin I_{(n-1)}^{(2)}$)}
      \STATE $skip(i)=skip(i-1)+1$
      \ENDIF
      \ENDFOR
        \STATE Skip the first $skip(C)$ chunks.
        \STATE $i^\prime=skip(C)+1$: the index of the first chunk to fetch
        \ENDIF

    \FOR{$i=i^\prime:C$}
    \STATE $j_k =deadline(i),k=\in \{1, 2\}$
    	\IF{($n=0$ or $i \in I_{n-1}^{(1)}$ or $i \in I_{n-1}^{(2)})$}
				\STATE ${B2}_{k}={B}_{k}, \forall k$, $t=deadline(i-1)$
      				\STATE $[B2_{k},\zeta_i^n(k)]=$ Backward$(k,i, j_k, \mathbf{B2}_{k}, Y_n, t) \forall k$
			\STATE $k_1=\arg\min (\zeta_i^n)$
      				\STATE $I_n^{(k_1)}=I_n^{(k_1)} \cup i$, $\mathbf{B}_{(k_1)}=B2_{(k_1)}$

	\ENDIF
      \ENDFOR
     \ENDFOR
   				\end{algorithmic}
				\caption{Offline No-Skip MP-SVC Algorithm }\label{algo:no_skip}
			\end{algorithm}
		\end{minipage}
	}
	\end{figure}

{\bf Online No Skip MP-SVC Algorithm}: In reality, we may not be able to predict the bandwidth for the entire video upfront, and thus all the stalls cannot be moved to the start. Based on the sliding window based approach described in Section \ref{sec:bw_err}, we bring the stalls in the window $W$ to the start of the window. After this change, the rest of the algorithm does similar adaptations to the offline schedules as in Section \ref{sec:bw_err}.





\section{Evaluation}
\label{sec:eval}
\subsection{Setup}

To evaluate the performance of the proposed algorithms, we have built an emulation testbed, which is depicted in Fig.~\ref{fig : emuSys}. The server, and the client nodes are both running in a Linux machine with kernel version 2.6.32, 24 CPU cores, and 32MB memory. The server and client communicate with each other via the loopback interface using the default TCP variant, TCP CUBIC~\cite{ha2008cubic}. Two orthogonal links were created to emulate the behavior of the multi-path streaming. To fetch a layer of a chunk, the client sends a request for that layer to the server via one of the two links, determined by our algorithm, thus the server will send back the data packets of that chunk layer via the same link. The emulation ends when all chunks are received. We record the received time, and the actual number of layers received for every chunk. To introduce bandwidth variations into the emulated experiments, we adopt \texttt{Dummynet}~\cite{dummynet} running on the client side, thus the incoming rates of the TCP packets from the server on the both links (i.e., 1 or 2) will be throttled according to the bandwidth datasets used. The bandwidths profile of the datasets for  both the links are described later in this section. Finally, we introduce a delay of 60ms between the server and the client. 

\begin{figure}
\centering
\includegraphics[trim=0in 2in 0in 0in, clip,width=0.4\textwidth, height=2in]{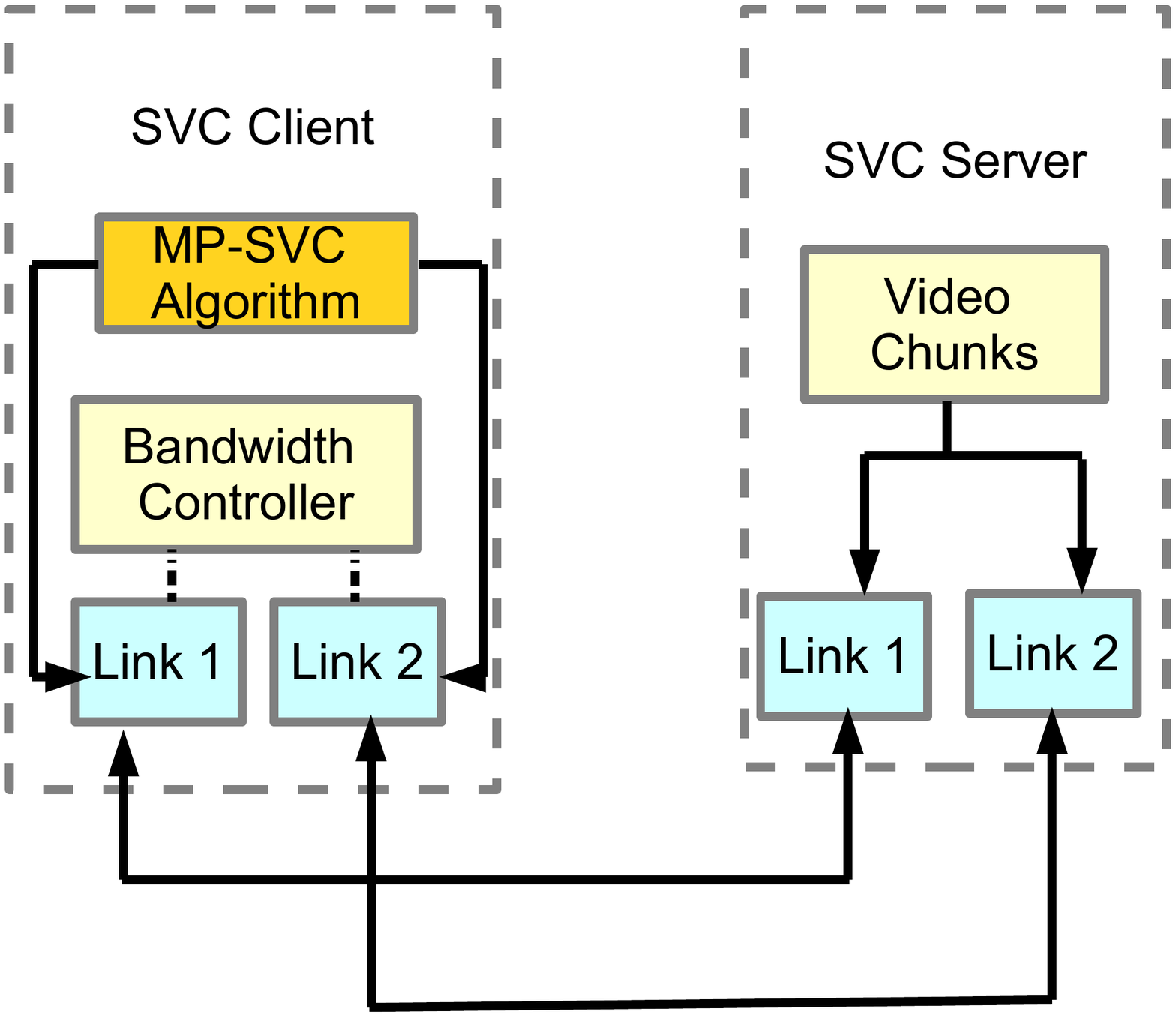}	
\vspace{-.3in}
 \caption{System Architecture for Emulation of MP-SVC.}
 \label{fig : emuSys}
 \end{figure}

{{\bf Video Parameters:} The video used for the evaluation is Big Buck Bunny, which is published in~\cite{SVCDataset}. It consists of 299 chunks (14315 frames), and the chunk duration is 2 seconds (48 frames and the frame rate of this video is 24fps).  The video is SVC encoded into one base layer and three enhancement layers. Table~\ref{tab : svc_rates} shows the cumulative nominal rates of each of the layers. The rates on the table are the nominal rates, and the exact rate of every chunk might be different since the video is VBR encoded.  In the table, ``BL'' and ``EL$_i$'' refer to the base layer and the cumulative (up to) $i$-th enhancement layer size, respectively. For example, the exact size of the $i$th enhancement layer is equal to EL$_i$-EL$_{(i-1)}$}. 

{{\bf Bandwidth Traces:} For bandwidth traces, we used two public datasets, representing the two paths. The first dataset (denoted Dataset1)  consists of continuous 1-second measurement of throughput  of a moving device in Telenor's mobile network in Norway \cite{riiser2013commute}. The second dataset (denoted Dataset2), is the  FCC dataset \cite{fccData}, which consists of more than 1 million sets of throughput measurements. Both these datasets have been post-processed in \cite{yin2015control} to give 1000 traces, each of 6-minute length which will be used in this paper for evaluations.}

\begin{figure}
\input{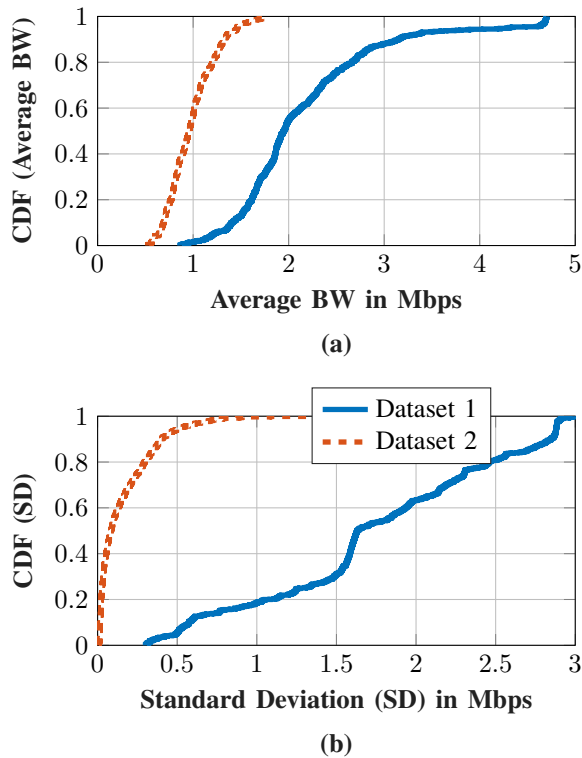}
 \vspace{-.1in}
 \caption{Statistics of the two bandwidth traces: (a) mean, and (b) standard deviation of each trace's available bandwidth.}
 \label{fig : bwStatV2}
 \vspace{-.1in}
 \end{figure}

\begin{table}[htb]
	  \vspace{-.1in}
  \centering
  \caption{SVC encoding bitrates used in our evaluation}
  \begin{tabular}{|c|cccc|} \hline
    playback layer & BL & EL1 & EL2 & EL3 \\ \hline
   {nominal Cumulative} rate (Mbps) & 0.6 & 0.99 & 1.5 & 2.075 \\ \hline
  \end{tabular}
  \label{tab : svc_rates}
  \vspace{-.1in}
\end{table}

Dataset1 has higher average bandwidth than Dataset2, but it also has higher variance as it can be seen in Fig~\ref{fig : bwStatV2}(a-b). We use Dataset1 as traces for the first link while Dataset2 as traces for the second link.

 {\bf Algorithm Parameters} For the online version of the proposed algorithm (MP-SVC, and MP-Pref-SVC, described in Section \ref{sec:bw_err} and \ref{sec:pref}), we set our algorithm's parameters as follows. We choose $W=10$ chunks, $B_{min}=4$ seconds, and $B_{max}=2$ minutes (60 chunks). We tried different buffer sizes, and the comparisons between different algorithms were qualitatively the same. Moreover, we choose $\alpha=2$ seconds.  We use the harmonic mean of the last $10$ seconds ($\beta=10$ seconds) to predict the future bandwidth. In other words, the predicted bandwidth for the entire window is set to the value of the harmonic mean of the last 10 seconds (or less in the start, when less data is available). Since there is no prediction at the beginning, the first two chunks are fetched at base layer quality where the first chunk is fetched over the first link, and the second one is fetched over the second link. For the link $1$ preference case, we assume $n_2=0$. Thus, link $2$ is only used to avoid skips, and not to increase the quality of chunks beyond the base layer. Similarly for the No-Skip version, the key use of link $2$ is to avoid/reduce stalls while not to improve further quality beyond the base layer. The proposed online algorithms without and with link $1$ preference  (MP-SVC, and MP-Pref-SVC, respectively) for both skip and no-skip based streaming scenarios are compared with the following baseline algorithms.


{\bf Multi-path, and SVC version of  Buffer-based Approach (MP-BBA)}: The authors of \cite{BBA} proposed a buffer-based algorithm, BBA,  for a single path and No-Skip streaming. BBA adjusts the streaming quality based on the playback buffer occupancy.
Specifically, the quality depends on two thresholds. If the buffer occupancy is lower (higher) than the lower (higher) threshold, chunks are fetched at the lowest (highest) quality.  If the buffer occupancy lies in between the two thresholds, the buffer-rate relationship is determined by lower thresholding the linear function between these extreme points to the available quality levels. We use the $30$ and $90$ seconds as the lower, and upper thresholds on the buffer length respectively. However, the standard BBA algorithm is a single path algorithm. In our proposed variant, MP-BBA, the quality is decided as per the BBA algorithm. In order to split layers across the two paths, the layers are split using the choice that minimizes the completion time of the chunk. Finally, we adapt this algorithm to the skip version where  a chunk  not  downloaded before its deadline  is skipped. The same modification to skip version is used for other no-skip versions described later. 


{\bf Multi-path version of Buffer-based Approach using Multipath TCP (MPTCP-BBA)}: This algorithm decides the quality based on BBA algorithm as described in MP-BBA. However, these quality decisions are fetched using Multi-path TCP and the decision of splitting into paths is no longer needed.

{\bf MPTCP-Pref-BBA \cite{han2016mp}}: The authors of \cite{han2016mp}  proposed Multi-path Dash (MP-DASH) which is an Adaptive MPTCP Video Streaming Over
Preference-Aware multi-path that takes the chunk quality decision based on BBA algorithm \cite{BBA}. Even though \cite{han2016mp} used AVC, the same decisions can be used for SVC.

{\bf MSPlayer~\cite{chen2016msplayer}}: 
MSPlayer takes quality decisions of the next two chunks based on the bandwidth prediction. Odd chunks are fetched on the first link, while even chunks are fetched on the second link.  If the current bandwidth measurement of the slow link (link with lower predicted throughput) is larger than ($1 + \delta$) times the estimated value, the chunk size is doubled and rounded to the nearest feasible chunk size. Similarly, if the current value is less than ($1 - \delta$) times  the estimated value, the chunk size is halved and rounded to the nearest  chunk size. The size of the chunk which is to be downloaded using the fast link is adjusted based on the throughput ratio. In other words, its size is equal to the size of the chunk fetched over the slow link times the ratio of the predicted bandwidths of the two links. Finally, it is rounded to the nearest  chunk size. We set $\delta$ to its default value ($5\%$) \cite{chen2016msplayer}.  We compare the proposed algorithms with MSPlayer only in case of equal priority since there is no known MSPlayer version where one link is more preferred over the other.

{\bf MPTCP-Festive}: We use the default setting of Festive algorithm as described in \cite{jiang2012improving} over the two links combined using multi-path TCP.

{\bf MPTCP-Pref-Festive \cite{han2016mp}}: The authors of \cite{han2016mp} also considered a Preference-Aware algorithm based on Festive algorithm \cite{jiang2012improving}. The decisions can be used over SVC.



{\bf MPTCP-SVC and MPTCP-Pref-SVC Algorithms}: These are the variants of the online MP-SVC and MP-Pref-SVC algorithm, where Multipath TCP can be used, as described in Section \ref{sec:bw_err}. 


{\bf off-MP-SVC and off-MP-Pref-SVC Algorithms}: These are the offline MP-SVC and MP-Pref-SVC as described in Section \ref{skipalgo} and their No-Skip variants described in section \ref{no_skip} where a genie-aided perfect bandwidth prediction is known for the entire video duration thus forming an upper bound of the performance (as measured by the proposed objective) of any online algorithm that splits layers on two paths. 


{\bf off-MPTCP-SVC and off-MPTCP-Pref-SVC Algorithms}: These algorithms run MPTCP-SVC with a genie-aided perfect bandwidth prediction is known for the entire video duration thus forming an upper bound of the performance (as measured by the proposed objective) of any online algorithm which may or may not use Multipath TCP.
 



\begin{figure}
\centering
\includegraphics[trim=0.3in 0in .1in 0in, clip,  width=.48\textwidth]{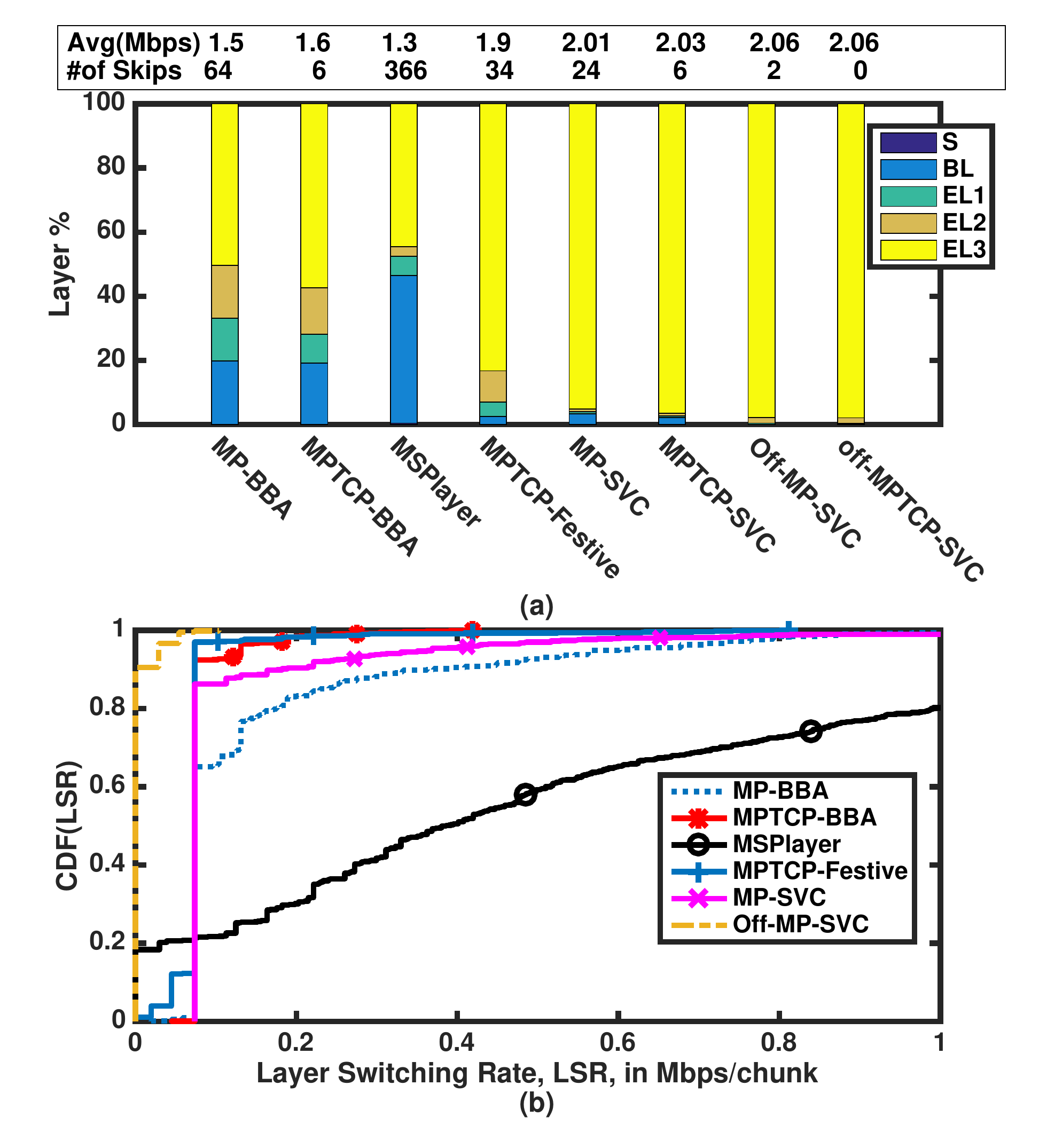}	
 \vspace{-.3in}
 \caption{ Skip based Streaming without link preference: (a) layer distribution, and (b) CDF of layer switching rate.}
 \label{fig : BCompeq}
\end{figure}

\begin{figure}
\centering
\includegraphics[trim=0.3in 0.1in .1in 0in, clip,  width=.48\textwidth]{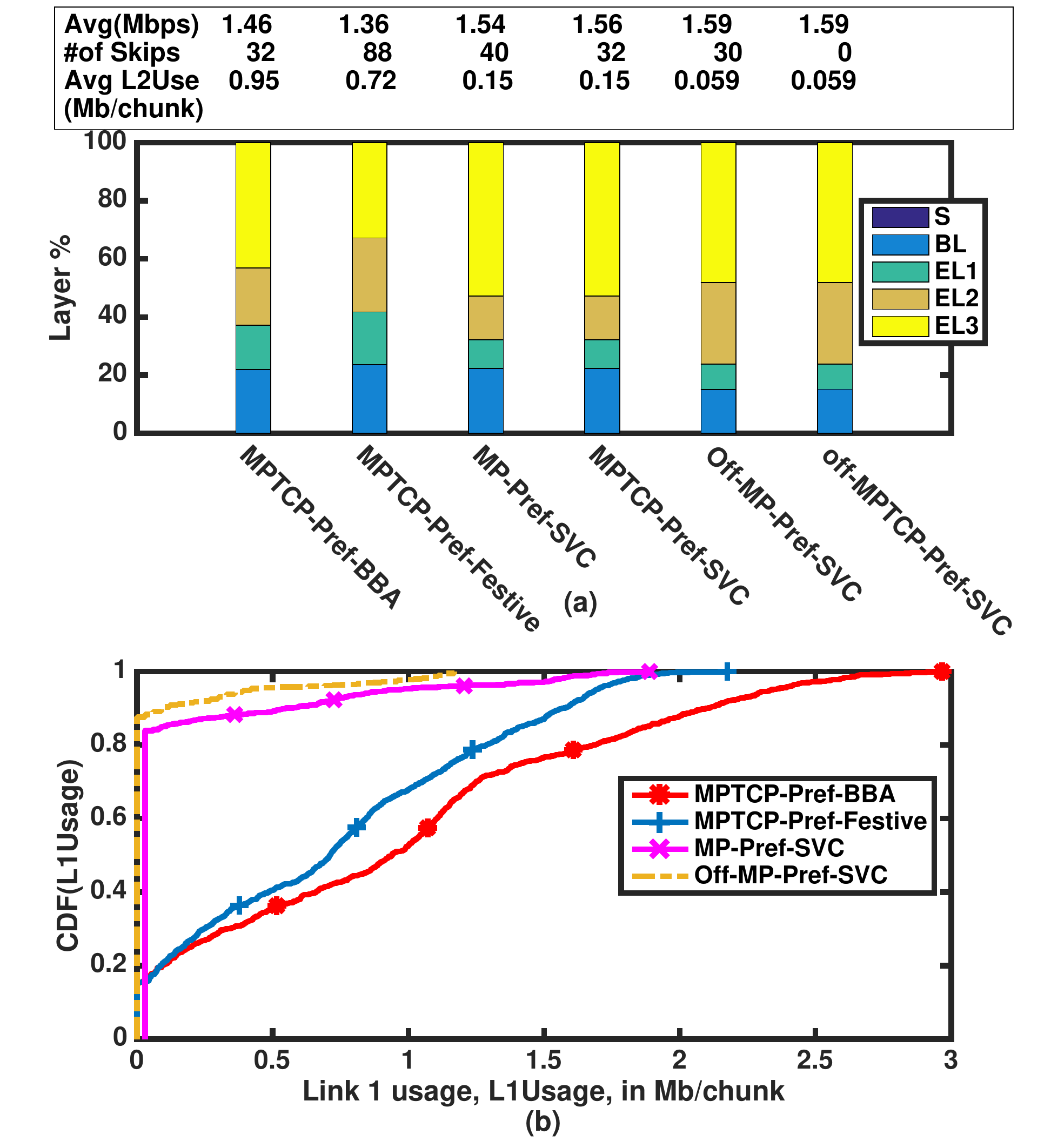}	
 \caption{ Skip based streaming with Link Preference: (a) layer distribution,  and (b) CDF of link 1 usage.}
 \label{fig : BComppref}
\end{figure}

\subsection{Skip Based MP-SVC Algorithm Without Preference}
\label{sec:eval_skip}



%
In this subsection, we will evaluate MP-SVC 
with a comparison to the baseline approaches in the skip based scenario. The results are shown in Fig.~\ref{fig : BCompeq}. The startup delay is chosen to be 5 seconds. Both the off-MP-SVC and off-MPTCP-SVC represent the optimal fetching policy without and with use of multi-path TCP and  thus represent the best possible strategies for the objective.


Fig.~\ref{fig : BCompeq}-a  shows the probability mass function of the number of chunks fetched at the different qualities (S=skips, corresponding to the chunks which were not fetched). The average playback rate among all traces, and the total number of skipped chunks among all videos is displayed on the top of Fig.~\ref{fig : BCompeq}-a. We first see that MP-SVC significantly outperform the algorithms that do not use multi-path TCP (MP-BBA, and MSPlayer). For instance, about $50\%$, and $45\%$ of the chunks are fetched at the $E_3$ quality in MP-BBA, and MSPlayer respectively. In contrast, the proposed algorithm MP-SVC fetches about $95\%$ of the chunks at the $E_3$ quality.  Further, MSPlayer runs into highest number of skips  ($366$ chunks were skipped in total, which corresponds to $12$ minutes of stall duration) while MP-BBA has 64 skipped chunks. In contrast, MP-SVC only has 24 skips, corresponding to less than $1$ minutes of video playback. Another thing to note is that though MP-SVC is less flexible than MPTCP-SVC, MPTCP-BBA, and MPTCP-Festive due to the restriction of choosing a layer on only one of the links, yet, the Quality of Experience obtained in the MP-SVC is almost the same as compared to the MPTCP-SVC. Moreover, MP-SVC is  better than MPTCP-Festive both in terms of avoiding skips and achieving higher average playback rate, and it outperforms MPTCP-BBA in the average quality (achieves $26\%$ higher) with  slightly higher number of skips. MP-SVC incorporates  bandwidth prediction and the deadline of the chunks into its optimization based decisions, prioritizes the later chunks, and re-considers the decisions every 2 seconds. These properties help MP-SVC be adaptive to different bandwidth regimes and variations in the bandwidth profiles.


  
 Fig.~\ref{fig : BCompeq}-b  shows  the distribution of the layer switching rates (LSR), which is defined as: $$\frac{1}{M}\sum_{i=2}^C |E\{X(i)\}-EX\{(i-1)\}|$$ where $C$ is number of chunks and $E\{X(i)\}$ is the expected play-back rate of $i$-th chunk, so if the $i$-th chunk is fetched at $n$-th enhancement layer quality, $E\{X(i)\}$ will be equal to the nominal cumulative rate of the $n$-th enhancement layer.  Intuitively, LSR quantifies the frequency of the layer change across the chunks and should be low for better quality of experience (QoE). MSPlayer performs poorly in terms of switching rate as compared to the other algorithms since it does not consider the buffer length and  doubles or halves the quality based on the predicted bandwidth. MP-BBA also have higher switching rate. MPTCP-Festive is slightly better in terms of switching rate, while worse in the average quality and the skip durations.


 Finally, we compare the online algorithms MP-SVC and MPTCP-SVC with the offline versions of them. We note from Fig.~\ref{fig : BCompeq} that although, MP-SVC and MPTCP-SVC use prediction window of only 10 chunks (short window), and  harmonic mean of the last $\beta=10$ seconds for predicting future bandwidth, they achieve a fetching policy that is very close to the one achieved by the offline algorithm (perfect prediction for entire video duration) with slightly  higher switching rate and a small increase in number  of  skips. The slight increase in switching rate for online schemes is in part since the first few chunks are downloaded in the base layer (because there is no bandwidth prediction initially), and a single jump from base layer to $EL3$ contributes $1475/80$ kbps/chunk, which is $\approx .02$ Mbps/chunk, to the LSR metric.  Prioritizing later chunks, re-considering the decisions every 2 seconds, and adjusting to prediction error, all help reducing the gap between the online and the offline  algorithm.

\subsection{Skip Based MP-SVC Algorithm with Link Preference }
\label{sec:eval_skip_pref}
In this subsection, we will evaluate the proposed algorithms, MP-Pref-SVC and MPTCP-Pref-SVC, and compare them to  the link preference based baseline approaches in the skip based scenario. The results are shown in Fig.~\ref{fig : BComppref}. The startup delay is chosen to be 5 seconds. The algorithms off-MP-Pref-SVC and off-MPTCP-Pref-SVC represent the  fetching policies when the bandwidth profiles are known non-causally without and with utilization of multi-path TCP, respectively.  

Fig.~\ref{fig : BComppref}-a  shows the probability mass function of the number of chunks fetched at the different qualities. We first note that there is no algorithm that fetches a chunk from a single path to compare with. Thus, we compare the algorithms with the versions that use multi-path TCP even though that is a disadvantage to MP-Pref-SVC since it does not use multi-path TCP. We see that MP-Pref-SVC achieves $15\%$ higher average playback rate than MPTCP-Pref-Festive with less number of skips, and about $6\%$ higher average playback rate than MPTCP-Pref-BBA with slightly higher number of skips. On the other hand, MPTCP-Pref-SVC algorithm outperforms both MPTCP-BBA and MPTCP-Festive both in terms of the achievable average playback rate and the number of skips. However, when one of the  link is preferred, not only the average playback rate is important, but the bandwidth usage of the less preferable link (link 2) is also important. 
%
%
Thus, in Fig.~\ref{fig : BComppref}-b, we plot the CDF of link 2 usage (L2Usage), and the total amount of link 2 usage per chunk is also displayed above the figure. We see that both MP-Pref-SVC and MPTCP-Pref-SVC  use lower amount of link $2$'s bandwidth as compared to MPTCP-BBA and MPTCP-Festive. In about $80\%$ of the bandwidth traces, MP-Pref-SVC and MPTCP-Pref-SVC used link $2$ to fetch only one chunk, and that chunk is necessary at the beginning to predict the  available bandwidth of link 2 (the offline algorithms did not use link $1$ for $80\%$ of the traces). However, both MPTCP-BBA and MPTCP-Festive used link $2$ to fetch more than 1Mb/chunk for $50\%$ of the bandwidth traces. Thus, the proposed algorithms give comparable or better average qualities than the baselines with significantly lower utilization of link $2$.

 \subsection{No-Skip  MP-SVC and MP-Pref-SVC Algorithms}
 \begin{figure}
\centering
\includegraphics[trim=0in 0in .1in 0in, clip,  width=.48\textwidth]{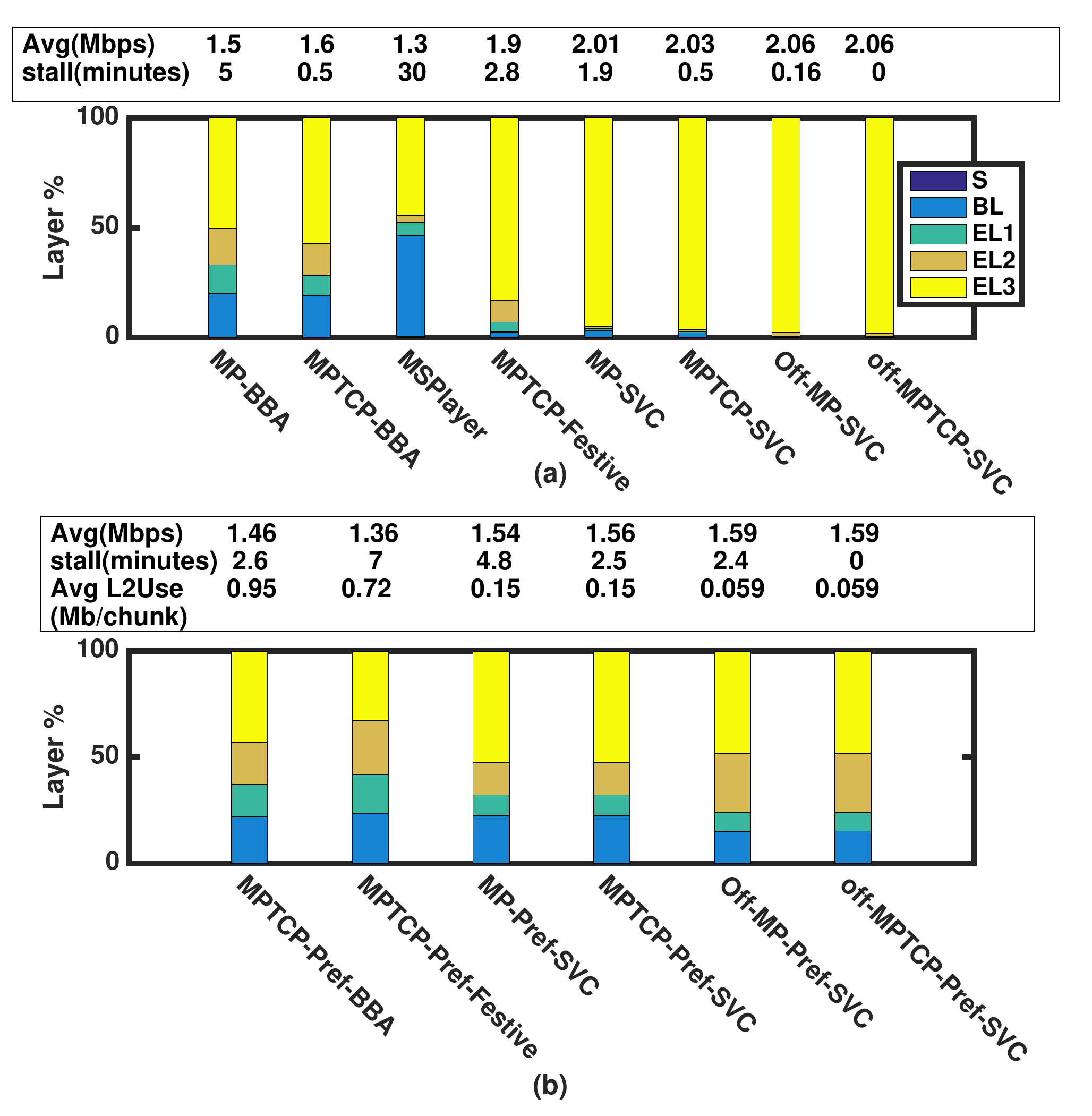}	
 \caption{(a) No-Skip based streaming without link preference,  (b)  No-Skip based streaming with link preference.}
 \label{fig : noSkipComp}
\end{figure}

In this subsection, we evaluate our proposed algorithms for the No-Skip based scenario. The initial startup delay is chosen to be the minimum of 5s and the download time of the first chunk while the other parameters are the same as in the skip-based version. The results are shown in Fig.~\ref{fig : noSkipComp}. The average playback qualities and the overall stall durations over all videos are shown above the sub-figures.  The first thing to note is that the layer distribution results are very similar to that in the skip based scenario for all the algorithms. The main difference in the no-skip version is that since there are no skips, the video may have multiple stalls. Thus, in Fig.~\ref{fig : noSkipComp}, we give the duration of stalls (re-buffering time) rather than skips. We omit the c.d.f. of the layer switching rate  and the LTE usage figures since the results are qualitatively similar to that in the skip version. We however give the average link $1$ usage above Fig. \ref{fig : noSkipComp}(b).  We see that   MPTCP-SVC and MPTCP-Pref-SVC  outperform all other baseline algorithms in all the considered metrics both when the link $1$ is preferred or when both the links have similar preference. Further, MP-SVC and MP-Pref-SVC has slightly higher stalls than the MPTCP variants with the other metrics being very close. It can also be seen that MP-SVC  outperforms the baseline algorithms that do not use MPTCP (MP-BBA and MSPlayer) in all the considered metrics.








\section{Conclusions}\label{sec:concl}

This paper provides a preference-aware adaptive streaming algorithm for a video encoded with Scalable Video Coding (SVC) over multiple paths, where each layer of every chunk can be fetched from only one of the paths. The problem of optimizing the user's quality of experience (QoE) subjected to the available bandwidth, chunk's deadlines, and link preference is formulated as a non convex optimization problem. It is shown that this non convex problem can be solved optimally with a complexity that is polynomial in the video length for some practical and special cases. Further, an online algorithm is proposed where 
several challenges including  bandwidth prediction
errors are addressed. Extensive emulations with real traces of public dataset reveal the robustness of our schemes and demonstrate their significant
performance improvement compared to other state-of-the-art multi-path algorithms.

\bibliographystyle{IEEEtran} %
\bibliography{refs,bib,multipath_feng}


\appendices
\section{Proof of Results in Section 3.4} \label{proof_skip}

\subsection{Proof of Lemma \ref{them:skip1}}
\label{apdx_skip1}
The forward scan for base layers decides to skip a base layer of a chunk $i$ only if the total bandwidth up to a deadline of a chunk $j \geq i$ is not enough to fetch all chunks $1$ to $j$. Since, the bandwidth up to the deadline of the $j^{th}$ chunk is not enough to fetch all chunks $1$ to $j$. Any other feasible algorithm, \ie algorithm that does not violate the bandwidth constraint, must have a skip for a chunk $i^\prime \leq j$. Thus, for every base layer skip of the proposed algorithm, there must be a base layer skip or more for any other feasible algorithm. The forward scan repeats for every enhancement layer in order with the remaining bandwidth. For every layer $n$, the proposed algorithm decides to skip the $n^{th}$ layer of a chunk $i$ in two cases:
\begin{itemize}
\item If the $(n-1)^{th}$ layer of this chunk is not decided to be fetched. Any feasible algorithm has to skip the $n^{th}$ layer of a chunk if its $(n-1)^{th}$ layer is not decided to be fetched; otherwise constraint~(\ref{equ:c2eq1}) will be violated.

\item The remaining bandwidth after excluding whatever reserved to fetch layers $1$ to $n-1$ for all chunks is not enough to fetch the $n^{th}$ layer of a chunk $j \geq i$. Thus, any feasible algorithm will decide to skip an $n^{th}$ layer of a chunk $i^\prime < j$. Otherwise the bandwidth constraint will be violated
\end{itemize}

Therefore, for any $n^{th}$ layer skip of the proposed algorithm, there must be an $n^{th}$ layer skip or more for any feasible algorithm. Thus, the algorithm achieves the minimum number of skips and that concludes the proof.

\subsection{Proof of Lemma \ref{them:skip2}}

We note that the proposed algorithm brings all $n$th layer skips to the very beginning (if necessary, it skips the earliest ones). We note from Section \ref{apdx_skip1} that if the ordered set of
$n$th layer skips for the backward algorithm is $i_1, i_2, \cdots, i_H$ and for any feasible algorithm with same number of $n$th layer skips is $j_1, j_2, \cdots, j_H$, then $i_k \le j_k$ for  all $k = 1, \cdots, H$.
Earlier $n$th layer skips provide bandwidth that can be available for more of the future chunks than any other choice since it comes before the deadline of more chunks thus proving the result in the statement of the lemma. 

\section{Proof of the optimality of Avoid-Skips MP-SVC}\label{apdx_pref}
\subsection{Proof of Lemma \ref{lem:avoid_skips} }\label{proofPref0}
From theorem \ref{them:skip1} in section \ref{skipalgoEqual}, we know that MP-SVC algorithm achieves the minimum number of base layer skips. However, Pref MP-SVC does not fetch less base layers, it only post processes the MP-SVC decisions in order to move as much as possible of base layers that have been decided to be fetched over link 2 to link 1. Therefore, the Pref MP-SVC will never fetch less number of base layers than MP-SVC and that concludes the proof.

\subsection{Proof of Lemma \ref{lem:avoid_skips2} }\label{proofPref1}
If there are two algorithms (1 and 2) achieve the same number of base layers and algorithm 2 fetches $\Delta$ more base layers than algorithm 1 by link 2. Then, the same $\Delta$ chunks should have been fetched by link 1. Let's denote the number of base layers fetched by link 1 in algorithm 2 by $H_1$, and the number of base layers fetched by link 2 in algorithm 2 by $H_2$. Hence, the number of base layers fetched by link 1 using algorithm 1 is $H_1+\Delta$, and the number of base layers fetched by link 2 in algorithm 1 is $H_2-\Delta$. Moreover, let's denote the achieved objective for the higher layers of algorithm 1 by $D_1$ and algorithm 2 by $D_2$ where $D_2 > D_1$. Therefore, the total objective of algorithm $1$ is:
\begin{equation}
\lambda_0^1(H_1+\Delta)+\lambda_0^2(H_2-\Delta)+D_1
\end{equation}
and the total objective of algorithm $2$ is:
\begin{equation}
\lambda_0^1H_1+\lambda_0^2H_2+D_2
\end{equation}

However, when $\lambda$s satisfy (\ref{basic_gamma_1_0}) the following holds true::
\begin{equation}
\lambda_0^1\Delta > \lambda_0^1H_1+\lambda_0^2H_2+D_2 
\end{equation}

Therefore, 
\begin{equation}
\lambda_0^1H_1+\lambda_0^2H_2+D_2 < \lambda_0^1(H_1+\Delta)+\lambda_0^2(H_2-\Delta)+D_1  
\end{equation}
and that concludes the proof.
\subsection{Proof of Lemma \ref{lem:avoid_skips3} }\label{proofPref2}

Pref MP-SVC re-runs MP-SVC for link 1 only and for chunks that were initially decided to be fetched by link 2 using the remaining bandwidth of link 1. From theorem \ref{them:skip1} in section \ref{skipalgoEqual}, we know that MP-SVC will fetch the maximum number of base layers. Therefore, Pref MP-SVC moves the maximum number of layers from link 2 (less-preferable) to link 1 (more-preferable links). Hence minimum number of base layers will be fetched over link 2. 

\subsection{Proof of Lemma \ref{lem:avoid_skips4} }\label{proofPref3}
MP-SVC skips the earliest chunks if there are any. Exchange algorithm (algorithm 4) replaces any layers fetched over link 2 by earliest possible. Hence, all chunks fetched using lower priority links are the earliest possible. Therefore, if the ordered set of base layers that are decided to be fetched by link 2 for the Pref MP-SVC is $i_1, i_2, \cdots, i_H$ and for any feasible algorithm with same number of layers decided to be fetched by link 2 which are $j_1, j_2, \cdots, j_H$, then $i_k \le j_k$ for  all $k = 1, \cdots, H$.  We know that earlier skips and moving earlier layers to link 2 (less preferable link) provides available bandwidth for more chunks in link 1 since it comes before the deadline of more chunks. This proves the result as in the statement of the lemma since any other feasible algorithm with more number of  layer skips/movements will achieve smaller objective when $\lambda$s satisfies \eqref{basic_gamma_1_0} and \eqref{basic_gamma_1_1}, thus showing that it will not be optimal.

\subsection{Proof of Theorem \ref{theorem: theorem2}}\label{Prefproof}

Use of  Lemma \ref{lem:avoid_skips} shows that the proposed algorithm is optimal for base layer skips, and use of Lemma \ref{lem:avoid_skips3} shows that the algorithm achieves the minimum number of base layer skips with minimum usage of link 2 (less preferable link).  According to lemma \ref{lem:avoid_skips4}, running Avoid-Skips MP-SVC algorithm offers the maximum bandwidth per chunk for $E_1$ layer decisions among all feasible algorithms with same number of skips. Therefore, the bandwidth that is passed to $E_1$ decision is the maximum per chunk. The rest of the algorithm is a running of MP-SVC for link 1 and a link with the bandwidth being zero all times. MP-SVC is shown to be optimal in avoiding skips theorem \ref{them:skip1} and optimal in providing more bandwidth to every chunk by skipping the layers of the earliest chunks if there are any \ref{proofPref3}. Therefore, MP-SVC is achieving the optimal decision of the $n^{th}$ layer and passing the optimal bandwidth profile to the $(n+1)^{th}$ layer. That concludes the proof.

\end{document}